\documentclass[12pt,draft]{iopart}
 
\def\version{}

\usepackage{verbatim}
\usepackage{iopams}
\usepackage{nomencl}
\usepackage{ulem}

\usepackage{mathpazo}

\usepackage[mathscr]{eucal}
\usepackage{graphicx}
\usepackage{epsf}
\usepackage{color}

\usepackage{srcltx} 

\let\mysection\section
\renewcommand{\section}[1]{\setcounter{equation}{0}\mysection{#1}}

\newenvironment{proof}[1] {\vskip0.1cm\noindent{\rmfamily\itshape#1.}}{\hfill $\square$\vspace{0.15cm}}

\def\Ascr{\mathscr{A}}

\DeclareMathAlphabet{\mathpzc}{OT1}{pzc}{m}{it}

\newcommand{\abs}[1]{{\vert #1\vert}}

\def\bS{\boldsymbol{S}}

\def\be{\boldsymbol{e}}

\def\br{\boldsymbol{r}}
\def\bs{\boldsymbol{s}}

\def\bv{\boldsymbol{v}}

\def\emptyset{\varnothing} 
\def\a{\alpha} 
\def\b{\beta} 
\def\g{\gamma}
 
\def\d{\delta} 
\def\e{\varepsilon} 

\def\S{\Sigma}

\def\O{\Omega}

\def\p{\partial} 
\def\L{\Lambda}

\newfam\Bbbfam 
\font\tenBbb=msbm10 
\font\sevenBbb=msbm7 
\font\fiveBbb=msbm5 
\textfont\Bbbfam=\tenBbb 
\scriptfont\Bbbfam=\sevenBbb 
\scriptscriptfont\Bbbfam=\fiveBbb

\newcommand{\R}     {\mathbb{R}} 
\newcommand{\Z}     {\mathbb{Z}} 
 
\renewcommand{\P}   {\mathbb{P}} 
 
\newcommand{\E}     {\mathbb{E}} 
 
\newcommand{\T}     {\mathbb{T}}

\def\1{{\mathchoice {1\mskip-4mu\mathrm l}      
{1\mskip-4mu\mathrm l} 
{1\mskip-4.5mu\mathrm l} {1\mskip-5mu\mathrm l}}} 
 
\def\comment#1{} 
 
\def\myffrac#1#2 in #3{\raise 2.6pt\hbox{$#3 #1$}\mkern-1.5mu\raise 0.8pt\hbox{$#3/$}\mkern-1.1mu\lower 1.5pt\hbox{$#3 #2$}}
\newcommand{\ffrac}[2]{\mathchoice%
	{\myffrac{#1}{#2} in \scriptstyle}
	{\myffrac{#1}{#2} in \scriptstyle}
	{\myffrac{#1}{#2} in \scriptscriptstyle}
	{\myffrac{#1}{#2} in \scriptscriptstyle}
}

\newcommand\tfrac[2]{{\textstyle\frac{#1}{#2}}}
\newcommand\twoeqref[2]{(\ref{#1}--\ref{#2})}
 
\renewcommand{\theequation}{\thesection.\arabic{equation}} 
 
\newtheorem{theorem}{Theorem}[section] 
\newtheorem{lemma}[theorem]{Lemma} 
\newtheorem{prop}[theorem] {Proposition}

\renewcommand{\thefigure}{\arabic{figure}}

\renewcommand{\d}{{\rm d}}

\newcommand{\Exp}{\mathscr{E}\kern-0.2mm{\operatorname{xp}}}
\newcommand{\Log}{\mathscr{L}\kern-0.2mm{\operatorname{og}}}

\newcommand{\Acal}   {{\mathcal A }}
\newcommand{\Bcal}   {{\mathcal B }}

\newcommand{\Ecal}   {{\mathcal E }} 
 
\newcommand{\Gcal}   {{\mathcal G }} 
\newcommand{\Hcal}   {{\mathcal H }}

\newcommand{\Ocal}   {{\mathcal O }}

\newcommand{\Scal}   {{\mathcal S }}

\newcommand{\Ycal}   {{\mathcal Y }}

\def\Ascr{\mathscr{A}}

\def\BBE{\Bcal_{\rm E}}
\def\BBSW{\Bcal_{\rm SW}}

 \newcommand{\bt}{\boldsymbol{t}}
\newcommand{\bk}{\boldsymbol{k}}
\newcommand{\textd}{\mbox{d}}
\newcommand{\hate} {\be}
\newcommand{\hbS} {\hat{\bS}}
\newcommand{\hS} {{\hat S}}
\newcommand{\hHcal} {\hat{\Hcal}}
\newcommand{\eqref} {\eref}

\newcommand{\texte}{\mbox{\rm e}}
\newcommand{\texti}{\mbox{\small\rm i}}

\newcommand{\su}{\mathfrak{su}}

\renewcommand{\emph}[1]{\textit{#1}}

\begin{document} 
\title[\version\hfil 
Long-range order in the plaquette orbital model]
{True nature of long-range order  in a plaquette orbital model} 
 
\author{Marek Biskup$^{1,2}$ and Roman Koteck\'{y}$^{3,4}$}

\address{$^1$Department of Mathematics, UCLA, Los Angeles, CA 90024, U.S.A.}
\address{$^2$School of Economics, University of South Bohemia, 370 05 \v Cesk\'e Bud\v ejovice, Czech Rep.}
\address{$^3$Center for Theoretical Study, Charles University, 110 00 Praha 1, Czech Republic}
\address{$^4$Mathematics Department, University of Warwick, Coventry, CV4 7AL, United Kingdom}
\eads{\mailto{biskup@math.ucla.edu}, \mailto{R.Kotecky@warwick.ac.uk}}

\begin{abstract}
\noindent
We analyze the classical version of a plaquette orbital model that was recently introduced and studied numerically by S.~Wenzel and W.~Janke. In this model, edges of the square lattice are partitioned into $x$ and $z$-types that alternate along both coordinate directions and thus arrange into a checkerboard pattern of $x$ and $z$-plaquettes; classical $O(2)$-spins are then coupled ferromagnetically via their first components over the $x$-edges and via their second components over the $z$-edges. 
We prove from first principles that, at sufficiently low temperatures, the model exhibits orientational long-range order (OLRO) in one of the two principal lattice directions. Magnetic order is precluded by the underlying symmetries. A similar set of results is inferred also for quantum systems with large spin although the \hbox{spin-$\ffrac12$} instance currently seems beyond the reach of rigorous methods.
We point out that the Ne\'el order in the plaquette energy distribution observed in numerical simulations is an artefact of the OLRO and a judicious choice of the plaquette energies. In particular, this order seems to disappear when the plaquette energies are adjusted to vanish at the ground-state level. We also discuss the specific role of the underlying symmetries in Wenzel and Janke's simulations and propose an enhanced method of numerical sampling that could in principle significantly increase the speed of convergence.

\end{abstract}


\vspace{2pc}

\noindent
To appear in: \textit{Journal of Statistical Mechanics: Theory and Experiment} 
\maketitle


\section{Introduction}

\noindent
The physics of transition-metal compounds~\cite{Goodenough:1955p4961} and of the underlying effects such as colossal magnetoresistance~\cite{BKK} has recently spawned a number of spin models of high theoretical and practical interest~\cite{KIKugel:1976p4966,vandenBrink:1999p4960}. Among the common features of these models is that (1)~the degrees of freedom are represented by quantum or classical spins residing at the vertices of a regular lattice, typically, the square or cubic lattice, (2)~the interactions are nearest neighbor and ferromagnetic but (3)~only certain components   ---   or projections   ---   of the spins are coupled over each edge of the lattice. The spin variables actually represent effective degrees of freedom (pseudospins); typically, occupation characteristics of a partially filled atom orbital. The interaction is also effective and it is arrived at by considerations of crystal fields mediated by interlaced atoms~\cite{KM}, or by appealing to Jahn-Teller distortions~\cite{Jahn:1937p4959} 
(or both methods, via different routes,~\cite{You:2007p4908}). Systems of a similar nature have sprung up independently in the field of topological quantum computation, e.g., the Kitaev model~\cite{Kitaev:2003p4963,Kitaev:2006p4967}, and so besides practical incentives to develop a theory for the behavior of these models, there are also strong theoretical reasons to understand their possible technical implementations~\cite{DSV,
Jetal}.

The definition of the aforementioned class of models starts by partitioning all edges of the lattice into families indexed by some~$\alpha$; a generic edge in the $\alpha$-th family is then denoted by $\langle\br,\br'\rangle_\alpha$. The Hamiltonian invariably takes the form
\begin{equation}
\label{E:Ham}
\mathcal H:=-\sum_\alpha\sum_{\langle \br,\br'\rangle_\alpha} J_\alpha\, S_{\br}^{(\alpha)}S_{\br'}^{(\alpha)}
\end{equation}
with positive coupling constants, $J_\alpha>0$. The explicit meaning of the projections $S_{\br}^{(\alpha)}$ is then a matter of what specific model one wishes to consider. 

Two examples of interest have been studied earlier: the \emph{orbital compass model} (e.g., 
\cite{Khaliullin,Mishra:2004p4919,Dorier:2005p4928,Nussinov:2005p4929,Orus:2009p4933,Brzezicki:2010p4915}), where $S_{\br}^{(\alpha)}$, $\alpha=1,2,3$, are the corresponding Cartesian components of the quantum spin and $\langle \br,\br'\rangle_\alpha$ is an edge in the $\alpha$-th lattice direction, and the \emph{120-degree model} 
(e.g.,~\cite{Nussinov:2004p435,Biskup:2005p701,Biskup-Chayes-Starr}), where the meaning of the edge $\langle \br,\br'\rangle_\alpha$ is preserved but $S_{\br}^{(\alpha)}$, $\alpha=1,2,3$, now denotes the projections of the (three-component) spin~$\bS_{\br}$ onto the vectors
\begin{equation}
\boldsymbol v_1:=\textstyle (1,0),\quad \boldsymbol v_2:=\bigl(-\frac12,\frac{\sqrt3}2\bigr),\quad \boldsymbol v_3:=\bigl(-\frac12,-\frac{\sqrt3}2\bigr),
\end{equation}
i.e., $S_{\br}^{(\alpha)}:=\boldsymbol v_\alpha\cdot\bS_{\br}$.
Kitaev's model is defined similarly to the orbital compass model but the underlying graph is the honeycomb lattice. 

As usual, all models in the above class have a natural quantum version, where $\bS_{\br}$ is a three-component spin operator  ---   with a distinct irreducible matrix representation for each non-negative half integer   ---   and a classical version, where $\bS_{\br}$ is a vector \emph{a priori} uniformly distributed on the unit sphere in $\mathbb R^N$ (i.e., an $O(N)$-spin).

\subsection{The plaquette orbital model}
\noindent
Recently, an interesting variant of the orbital compass model has been proposed and studied by Wenzel and Janke~\cite{Wenzel:2009p4934}. Their model, which they termed the \emph{plaquette orbital model} (POM), is most naturally defined over the square lattice $\mathbb Z^2$, although generalizations to higher dimensions are straightforward. The index $\alpha$ takes only two values, $\alpha=1,2$, and the spin projections are defined as follows:
\begin{equation}
S_{\br}^{(1)}:=S_{\br}^x\quad\mbox{and}\quad S_{\br}^{(2)}:=S_{\br}^z,
\end{equation}
where, in agreement with Wenzel and Janke's notation, $(S_{\br}^x,S_{\br}^z)$ denote the Cartesian components of the vector-valued $O(2)$-spin $\bS_{\br}$ in the classical version while, in the quantum version, it denotes the corresponding pair of operators for the quantum spin.

\begin{figure}[ht]
\begin{center}
\epsfxsize=3.2in
\epsfbox{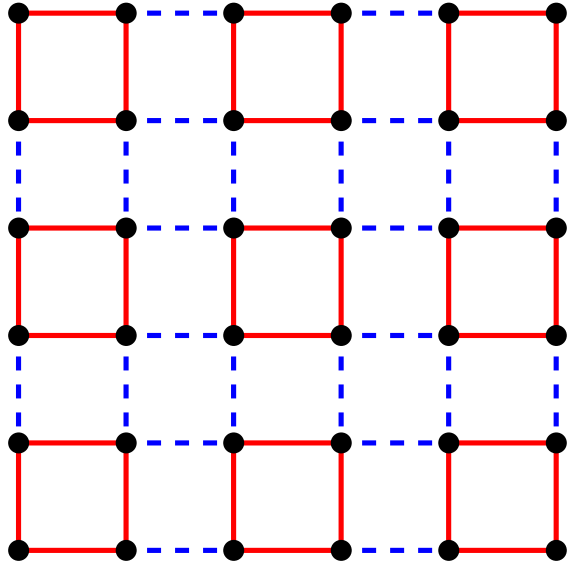} 
\end{center}
\noindent
\begin{minipage}{0.85\textwidth}
\small
\caption{The configuration underlying the definition of the plaquette orbital model. Here the $x$-components of the spins are coupled over the red (solid) edges and the $z$-components are coupled over the blue (dashed) edges.}
\end{minipage}
\end{figure}

What distinguishes this model from the previously discussed counterparts is the partitioning of the lattice: edges are designated as $x$ and $z$-types in an alternating fashion along each line of sites (copy of~$\mathbb Z$) inside $\mathbb Z^2$ so that all of the lattice faces (plaquettes) with the lower-left corner on the even sublattice of $\mathbb Z^2$ contain only edges of one type. Thus one gets the Hamiltonian in accord with the general form in Eq.~\eqref{E:Ham} with the edges of two types arranged into to a checkerboard pattern as in Fig.~1,
\begin{equation}
\label{E:1.4}
\mathcal H:=-J_1\sum_{\langle \br,\br'\rangle_x} \, S_{\br}^{x}S_{\br'}^{x}-J_2\sum_{\langle \br,\br'\rangle_z} \, S_{\br}^{z}S_{\br'}^{z}, 
\end{equation}
with the  $x$-components of the spins  coupled over the $x$-edges and $z$-components over the $z$-edges. The signs of the coupling constants are immaterial as they can always be absorbed into a sign change of the corresponding component on one of the sublattices.

In \cite{Wenzel:2009p4934}  Wenzel and Janke studied the POM numerically both in its classical and quantum version. For an order parameter they chose the plaquette energy,
\begin{equation}
\label{E:WJen}
\mathcal E_{\br}:=-\Bigl\{S_{\br}^{\alpha_1}S_{\br+\hate_1}^{\alpha_1}+S_{\br}^{\alpha_2}S_{\br+\hate_2}^{\alpha_2}+S_{\br+\hate_1}^{\alpha_3}S_{\br+\hate_1+\hate_2}^{\alpha_3}+S_{\br+\hate_2}^{\alpha_4}S_{\br+\hate_1+\hate_2}^{\alpha_4}\Bigr\},
\end{equation}
indexed by the position of the lower-left corner of the plaquette, where $\hate_1$ and $\hate_2$ are the unit vectors in the first and second coordinate direction and $\alpha_1,\dots,\alpha_4$ are either $x$ or $z$ depending on whether the corresponding edge is an $x$ or $z$-type, respectively. The simulations of \cite{Wenzel:2009p4934} indicated a clear onset of Ne\'el order in the plaquette energies at low temperatures for both the classical two-component ($O(2)$-spin) model, and the quantum, spin-$\ffrac12$ model. Explicitly, in one such state the plaquette energies on the $x$-plaquettes are low but those on the $z$-plaquettes are generally high and in another state the roles of $x$ and $z$-plaquettes are interchanged. (The mixed-type plaquettes end up with an intermediate energy in both circumstances.) However, no specific conclusion was attempted for other possible ordering mechanisms (magnetic order, orientational order, etc). 

\subsection{Goals and outline}
\noindent
In the present paper we wish to address the true nature of the phase transition in the POM by means of rigorous mathematical (i.e., analytic) methods that draw from earlier work on models of this kind~\cite{Nussinov:2004p435, Biskup:2005p701,Biskup:2004p697,Biskup-Chayes-Starr}. We will predominantly focus on the classical $O(2)$-spin version of the model as the quantum spin-$\ffrac12$ version poses technical difficulties that we do not yet know how to overcome. Nonetheless, thanks to the general theory~\cite{Biskup-Chayes-Starr}, the conclusions for the classical system permit straightforward extensions to quantum systems once the magnitude of the quantum spin (i.e., the quantity $\mathcal S$ such that $\hat{\bS}_{\br}\cdot \hat{\bS}_{\br}=\mathcal S(\mathcal S+1)$) is sufficiently large compared to the inverse temperature squared.

A key conceptual difference with the approach of Wenzel and Janke~\cite{Wenzel:2009p4934} is that instead of going via plaquette energies,  we directly attempt (and succeed in) proving \emph{orientational} long-range order (ORLO) of the spins in one of the two principal lattice directions. We then argue that the Ne\'el ordering in plaquette energies found in~\cite{Wenzel:2009p4934} is actually an \emph{artefact} of this ORLO: indeed, it is a direct consequence of the alignment of the spins along coordinate axes, the 2-periodicity of the interaction and the fact that the plaquette energies were not normalized to vanish in the ground states. In fact, with such normalization the Ne\'el order seems to disappear altogether.

The remainder of this paper is organized as follows: In Sect.~\ref{sec-II} we will discuss the ground states of the classical Hamiltonian and then state our principal results concerning the ORLO in the pure and diluted systems. In Sect.~\ref{sec-III}-\ref{sec-V} we provide full and reasonably self-contained mathematical proofs of these results. In Sect.~\ref{sec-VI} we discuss connections with the numerical findings and propose a potentially more efficient way to obtain samples of equilibrium configurations in this model. 


\section{Rigorous results}
\label{sec-II}

\subsection{Ground states}
Our discussion of the results opens up with the description of the ground states.
We will focus on the situation in finite volumes with periodic boundary conditions.
Specifically, let $\mathbb T_N$ be the $N\times N$ torus  ---  obtained by periodizing the square $\{0,\dots,N-1\}\times\{0,\dots,N-1\}$  ---  and
assume that $N$ is even to reflect the natural period-2 nature of the interaction. Let $\mathcal H_N$ denote the Hamiltonian on
$\mathbb T_N$ which we define by \eqref{E:1.4} with the edges $\langle\br,\br'\rangle_\alpha$ restricted to nearest neighbor pairs (of the proper type) on~$\T_N$. As is common, we call a configuration $\bS=(\bS_{\br})_{\br\in\mathbb T_N}$ a \emph{ground state} of $\mathcal H_N$ if
\begin{equation}
\mathcal H_N(\bS)=\min_{\bS^{\prime}} \mathcal H_N(\bS^{\prime}).
\end{equation}
 Here we note that the minimum is attained  ---  and a ground state exists  ---  by the sheer fact that $\mathcal H_N$ is a bounded 
 and continuous function on a compact metric space. The issue is how many ground states there are and how they can be concisely described.
 
 A common feature of all models of the type \eqref{E:Ham} is abundance of symmetries with respect to simultaneous flips of (specific) components of the spins. An indisputable advantage of the POM over the other systems is that here the flips can be made locally.
 Explicitly, for $\br$  with both coordinates even, i.e., the site designating a lower left corner of an $x$-plaquette, let~$\boldsymbol{\varphi}_{\br}(\bS)$ be the configuration defined by
 \begin{equation}
\bigl[\boldsymbol{\varphi}_{\br}(\bS)\bigr]_{\br^\prime}^x:=
\cases{-S_{\br^\prime}^x &for $\br^\prime=\br, \br+\hate_1,\br+\hate_2, \br+\hate_1+\hate_2$,\\
S_{\br^\prime}^x &otherwise,\\}
\end{equation}
and
\begin{equation}
\bigl[\boldsymbol{\varphi}_{\br}(\bS)\bigr]_{\br^\prime}^z:=S_{\br^\prime}^z.
\end{equation}
For $\br$  with both coordinates odd (lower left corner of a $z$-plaquette),
the map $\boldsymbol{\varphi}_{\br}$ is defined in a similar manner  ---  with the $z$-components reflected instead of the $x$-components. A moment's thought now shows that
\begin{equation}
\label{E:sym}
\mathcal H_N(\boldsymbol{\varphi}_{\br}(\bS))=\mathcal H_N(\bS),
\end{equation}
i.e., $\bS\mapsto\boldsymbol{\varphi}_{\br}(\bS)$ is a symmetry of the model.

\begin{theorem}
\label{T:ground}
Suppose $J_1=J_2>0$. Then every ground state of $\Hcal_N$ can be obtained from a constant configuration, $\bS_{\br}\equiv\hate$, for some unit vector $\hate\in\R^2$, by successive applications of a subset of the maps~$\bigl(\boldsymbol{\varphi}_{\br}\bigr)$.

If $J_1>J_2> 0$, then all ground states arise (via applications of (${\boldsymbol{\varphi}}_{\br}$)) from $\bS_{\br}\equiv\be_1$, while if 
$J_2>J_1> 0$, then all ground states arise (again, via ($\boldsymbol{\varphi}_{\br}$)) from $\bS_{\br}\equiv\be_2$.
\end{theorem}

This statement is a precursor of the orientational LRO that we will establish for low (but non-zero) temperatures. The key problem there will be the fact that the ground-state degeneracy significantly increases at the symmetry point $J_1=J_2$. This appears to be a common feature for all models covered by the equation \eqref{E:Ham}.

We also remark that in infinite volume (i.e., the model defined on all of $\Z^2$), the structure of ground states is considerably more complicated. (In infinite volume, a ground state is a configuration whose any local change will result in a non-negative change of energy.) 
It is still true that any configuration obtained from  constant configurations by means of the maps~$({\boldsymbol{\varphi}}_{\br})$
is a ground state. However, further ground states can be constructed by imposing linear (or other) interfaces.

\subsection{Orientational order}
We proceed to discuss our results for positive temperatures. Let $\beta:=\frac1{k_{\mbox{\tiny B}} T}$ denote the inverse temperature.
In the canonical ensemble, the spins on $\T_N$ are distributed according to the \emph{Gibbs measure} $\mu_{N,\b}$ that is defined by
\begin{equation}
\label{E:Gibbs}
\mu_{N,\b}(\d \bS):= \frac{\texte^{-\b \Hcal_N(\bS)}}{Z_{N,\b}}\prod_{\br\in\T_N}\nu(\d \bS_{\br}),
\end{equation}
where $\nu$ denotes the uniform (Haar) measure on the unit circle in $\R^2$ normalized, for later convenience, to $\sqrt{2\pi}$. The normalization constant $Z_{N,\b}$ is the partition function.

In order to formulate the existence of a long range order in a mathematically precise way, one often considers Gibbs measures directly in infinite volume  ---  i.e., 
as measures on infinite configurations~$\bigl(\bS_{\br}\bigr)_{\br\in\Z^2}$. These are defined by means of the so-called DLR-condition \cite{Dobrushin1,Dobrushin2,Lanford-Ruelle} stating that the conditional distribution in any finite $\L\subset\Z^2$ 
given $\bigl(\bS_{\br}\bigr)_{\br\in\Z^2\setminus\L}$ takes the above form with $\Hcal_N$ replaced by the Hamiltonian in $\L$
under the boundary conditions $\bigl(\bS_{\br}\bigr)_{\br\in\Z^2\setminus\L}$. 

A standard reference to the corresponding mathematical theory of these measures is Georgii~\cite{Georgii}. We will now list the aspects that have a bearing on our problem.
It is a standard fact that, for compactly-supported spins, infinite-volume Gibbs measures can be extracted (as weak limits) from sequences of finite-volume, or even torus, Gibbs measures.
In particular, estimates on $\mu_{N,\b}$ that hold uniformly in $N$ readily yield corresponding estimates for infinite volume limits of $\{\mu_{N,\b}\}$.

In light of the period-2 nature of the interactions, a (Gibbs) measure $\mu$ will be called \emph{translation-invariant} if
$\mu(\sigma_{\a}(\mathcal A)) =\mu(\mathcal A)$ for any event $\mathcal A$, where
\begin{equation}
\bigl[\sigma_{\a}(\bS)\bigr]_{\br}:=\bS_{\br-2 \be_{\a}}, \quad \a=1,2,
\end{equation}
are the shift operators.
This measure is said to be \emph{ergodic} if $\mu(\mathcal A)$ is either zero or one for any event $\mathcal A$ that is invariant under the translations above, i.e., such that $\sigma_{\a}^{-1}(\mathcal A)= \mathcal A$
for $\a=1,2$. Ergodic measures have the distinguished feature that all block averages converge to the corresponding expectations
(i.e., expected value represents the average value in sufficiently large boxes) and that a typical sample from these measures thus has a fairly homogeneous structure at sufficiently large spatial scales.

We begin by ruling out magnetic ordering at any parameters of the model.

\begin{theorem}
\label{T:no-magnet}
Let $J_1, J_2\ge0$ and  $\b\ge0$. Then
\begin{equation}
\E_\mu\bigl(\bS_{\br}\bigr)=0
\end{equation}
for all infinite-volume Gibbs states $\mu$.
\end{theorem}

The absence of magnetic order is of course a consequence of the symmetry $\bS\mapsto\boldsymbol{\varphi}_{\br}(\bS)$. This result is prototypical for all models of the kind \eqref{E:Ham}.
It does not, however, exclude the existence of an orientational long-range order.

\begin{theorem}
\label{T:LRO}
Suppose $J_1=J_2>0$. Then for each $\delta>0$, there is $\b_0=\b_0(\delta)$ and a sequence $({\e}_N)$ with $\e_N\to 0$ such that for each $\b\ge\b_0(\delta)$,
\begin{equation}
\label{E:musym}
\mu_{N,\b}\biggl(\,\sum_{\br\in\T_N}\bigl[S_{\br}^{\alpha}\bigr]^2\ge|\T_N|(1-\delta)\biggr)\ge \frac12-\e_N,\quad \a=x,z.
\end{equation}
In particular, for each $\b\ge\b_0(\delta)$, there exist two translation-invariant, ergodic infinite-volume Gibbs measures, $\mu_\b^{x}$ and $\mu_\b^{z}$,
such that
\begin{equation}
\label{E:ES^2x}
\E_{\mu_{\b}^{\alpha}}\bigl(\bigl[S_{\br}^{\alpha}\bigr]^2\bigr)\ge 1-\delta,\quad \a=x,z. 
\end{equation}
In addition, for any $\beta\ge\beta_0(\delta)$ and any translation-invariant, ergodic infinite-volume Gibbs state $\mu$ at inverse temperature $\beta$,
\begin{equation}
\label{E:ES^2}
\E_\mu\bigl(\bigl[S_{\br}^{\alpha}\bigr]^2\bigr)\in[0,\delta]\cup [1-\delta,1],\quad \a=x,z.
\end{equation}
\end{theorem}

Notice that \eqref{E:musym} states that typical configurations sampled from $\mu_{N,\b}$ on a large torus have a majority of all spins 
aligned either along direction $\pm\hate_1$ or direction $\pm\hate_2$
(with both orientations equally present thanks to the symmetries \eqref{E:sym}).
The inequality \eqref{E:ES^2x} is a version of this fact 
in
infinite volume and it manifestly demonstrates the occurrence of a phase transition.
The statement \eqref{E:ES^2} in turn implies that only the ground states oriented along the principal axes are stable under thermal perturbations. The infinite degeneracy at the symmetry point for $\beta=\infty$ is thus reduced to a two-fold degeneracy
\footnote[1]{A word of caution: Strictly speaking, the rigorous statement does not rule out the existence of additional extremal translation-invariant Gibbs states apart from those above.
However, an overwhelming majority of the spins in a typical configuration in any such state will be close either to $\pm\hate_1$ or to $\pm\hate_2$.} once $\beta<\infty$.
Away from the symmetry point we have the following:

\begin{theorem}
\label{T:away-from-symmetry}
If $J_1>J_2>0$ and  $\b\ge\b_0(\delta)$, then in all translation-invariant, ergodic infinite-volume Gibbs states $\mu$,
\begin{equation}
\E_\mu\bigl(\bigl[S_{\br}^{x}\bigr]^2\bigr)\ge 1-\delta, \quad\mbox{ for all } \br\in\Z^2.
\end{equation}
Similarly for $\E_\mu\bigl(\bigl[S_{\br}^{z}\bigr]^2\bigr)$ when $J_2>J_1>0$ and  $\b\ge\b_0(\delta)$.
\end{theorem}

The above results imply the existence of ORLO throughout the part of the quadrant in $(J_1,J_2)$-plane bounded away from zero. 

\subsection{Quantum systems}
While our methods currently seem unable to treat the quantum spin-$\ffrac12$ version of POM that 
was studied numerically in ~\cite{Wenzel:2009p4934}, the control of the classical system can be extended to the quantum problem with large spin. This is achieved essentially by plugging into the main result of \cite{Biskup-Chayes-Starr}. We proceed to introduce the technical aspects of the quantum POM that are necessary to state the relevant theorem.

In the quantum POM, the spins ${\hbS}_{\br}$ are\ three-component operators 
${\hbS}_{\br}=({\hS}_{\br}^x, {\hS}_{\br}^y,{\hS}_{\br}^z)$ satisfying the commutation rules of the Lie algebra $\su(2)$,
\begin{equation}
\bigl[\hS_{\br}^j,\hS_{\br^\prime}^k\bigr]=2\texti \hS_{\br}^\ell \delta_{\br,\br^\prime}
\end{equation}
for any cyclic permutation $(j,k,\ell)$ of $(x,y,z)$. We will work with the $(2\Scal+1)$-dimensional irreducible
representation of $\su(2)$, where $\Scal\in\{0,\frac12,1,\frac32,\dots\}$.
This representation is best discussed by means of the spin raising and lowering operators 
\begin{equation}
\hS_{\br}^{\pm}:= \hS_{\br}^x\pm\texti\hS_{\br}^y.
\end{equation}
The Hilbert space is the linear span of vectors 
\begin{equation}
|\dots M_{\br}\dots \rangle:= \bigotimes_{\br} |M_{\br}\rangle, \quad  M_r=-\Scal,-\Scal+1,\dots, \Scal-1,\Scal,
\end{equation}
on which the operators $\hS_{\br}^z, \hS_{\br}^+,\hS_{\br}^-$ act as follows
\begin{equation}
\label{E:Spm}
\eqalign{
\hS_{\br}^z|\dots M_{\br}\dots \rangle=M_{\br}|\dots M_{\br}\dots \rangle,\\
\hS_{\br}^+|\dots M_{\br}\dots \rangle=\sqrt{\Scal(\Scal+1)-M_{\br}(M_{\br}+1)}\,|\dots M_{\br}+1\dots \rangle,\\
\hS_{\br}^-|\dots M_{\br}\dots \rangle=\sqrt{\Scal(\Scal+1)-M_{\br}(M_{\br}-1)}\,|\dots M_{\br}-1\dots \rangle.}
\end{equation}
The Hamiltonian $\hHcal_N$ on the torus $\T_N$ is then the operator
\begin{equation}
\label{E:quantum-Ham}
\hHcal_N:=-J_1\Scal^{-2}\sum_{\langle \br,\br'\rangle_x} \, \hS_{\br}^{x}\hS_{\br'}^{x}-J_2\Scal^{-2}\sum_{\langle \br,\br'\rangle_z} \, \hS_{\br}^{z}\hS_{\br'}^{z}.
\end{equation}
The $y$-component of the spin does not enter the interaction. The scaling by $\Scal^{-2}$ ensures that the Hamiltonian is, for each~$N$, bounded uniformly in~$\Scal\ge\ffrac12$.

The thermodynamical equilibrium is described by means of a linear functional $\langle \boldsymbol{\cdot}\rangle_{N,\b}$ on the algebra $\Ascr_N$ of all bounded operators generated 
(via the spectral theorem) by the operators $\hS_{\br}^z, \hS_{\br}^+,\hS_{\br}^-$, $\br\in\T_N$.
Explicitly,
\begin{equation}
\langle \hat A\rangle_{N,\b}=\frac{\Tr(\hat A \texte^{-\b\hHcal_N})}{\Tr ( \texte^{-\b\hHcal_N})},\quad  \hat A\in\Ascr_N.
\end{equation}
Our main result for the quantum system is now as follows:

\begin{theorem}
\label{T:QLRO}
Suppose $J_1=J_2>0$. Then for each $\delta>0$, there are constants 
$c=c(\delta)>0$ and $\b_0=\b_0(\delta)<\infty$  such that for each $\b\ge\b_0(\delta)$ and $\Scal\ge c\b^2$, we have
\begin{equation}
\label{E:QLRO}
\Scal^{-4}\Bigl\langle\bigl[\hS_{\br}^{\alpha}\bigr]^2\bigl[\hS_{\br^\prime}^{\alpha}\bigr]^2\Bigr\rangle_{N,\b}\ge \frac12-\delta,\quad \a=x,z,
\end{equation}
while
\begin{equation}
\label{E:QLRO22}
\Scal^{-4}\Bigl\langle\bigl[\hS_{\br}^{y}\bigr]^2\bigl[\hS_{\br^\prime}^{y}\bigr]^2\Bigr\rangle_{N,\b}<\delta,
\end{equation}
and 
\begin{equation}
\label{E:QLRO13}
\Scal^{-4}\biggl|\Bigl\langle\bigl[\hS_{\br}^{x}\bigr]^2\bigl[\hS_{\br^\prime}^{z}\bigr]^2\Bigr\rangle_{N,\b}\biggr|<\delta,
\end{equation}
uniformly in $\br,\br^\prime\in\T_N$ provided $N$ is sufficiently large.
\end{theorem}

Note that the bound \eqref{E:QLRO13} yields the same bound on $\langle\bigl[\hS_{\br}^{z}\bigr]^2\bigl[\hS_{\br^\prime}^{x}\bigr]^2\rangle_{N,\b}$ because the identity $|\langle \hat A\hat B\rangle_{N,\beta}|=|\langle \hat B\hat A\rangle_{N,\beta}|$ holds for all self-adjoint operators $\hat A,\hat B$. It is standard that validity of such bounds implies nonanalyticity of the free energy in the appropriate conjugate variables. In our case, this will be the function
\begin{equation}
\!\!\!\!\!\!\!\!\!\!\!
f(h):=\lim_{N\to\infty}\frac1{N^d}\log\Tr\Biggl(\exp\biggl\{-\beta\hat{\mathcal H}_N+h\Scal^{-2}\sum_{\br\in\T_N}([S_{\br}^{x}]^2-[S_{\br}^{z}]^2)\biggr\}\Biggr),
\end{equation}
where the ``external field'' $h$ couples to the natural order parameter $[S_{\br}^{x}]^2-[S_{\br}^{z}]^2$. Based on this fact, we extract the corresponding result for the asymmetric situations as well:

\begin{theorem}
\label{T:QLRO2}
Suppose $J_1>J_2>0$ and let $c=c(\delta)>0$ and $\b_0=\b_0(\delta)<\infty$ be as above. Then for each $\delta$, each $\b\ge\b_0(\delta)$ and $\Scal\ge c\beta^2$, we have
\begin{equation}
\label{E:QLRO2}
\Scal^{-4}\Bigl\langle\bigl[\hS_{\br}^{x}\bigr]^2\bigl[\hS_{\br^\prime}^{x}\bigr]^2\Bigr\rangle_{N,\b}\ge 1-2\delta,
\end{equation}
while
\begin{equation}
\Scal^{-4}\Bigl\langle\bigl[\hS_{\br}^{\alpha}\bigr]^2\bigl[\hS_{\br^\prime}^{\alpha}\bigr]^2\Bigr\rangle_{N,\b}<\delta,\qquad\alpha=y,z,
\end{equation}
uniformly in $\br,\br^\prime\in\T_N$  for $N$ is sufficiently large. A similar result holds for $J_2>J_1>0$ with indices $x$ and $z$ interchanged.
\end{theorem}

We remark that the reason for assuming $\Scal\ge c\beta^2$ is that, in the underlying proof, we use thermal fluctuations to dominate the quantum fluctuations (namely, the effects resulting from the non-commutative nature of the relevant variables).

\section{Ground states}
\label{sec-III}
The goal of this section is to prove our characterization of the grounds states on $\T_N$. As a starting point we note the following rewrite of the energy function:

\begin{lemma}
\label{lemma-H-rewrite}
Let $\mathcal H_N$ denote the torus Hamiltonian. Then for all $\bS$,
\begin{equation}
\label{E:H-rewrite}
\hbox{\hglue-1cm}\mathcal H_N(\bS)=\frac{J_1}2\sum_{\langle \br,\br'\rangle_x} \, (S_{\br}^{x}-S_{\br'}^{x})^2
+\frac{J_2}2\sum_{\langle \br,\br'\rangle_z}  (S_{\br}^{z}-S_{\br'}^{z})^2-\sum_{\br}\bigl(J_1[S_{\br}^x]^2+J_2[S_{\br}^z]^2\bigr).
\end{equation}
\end{lemma}

\begin{proof}{Proof}
Every vertex in $\T_N$ has two $x$-edges and two $z$-edges coming out of it. Opening up the squares in the first two sums, the ``diagonal'' terms there are easily checked to exactly cancel the terms in the third sum.
\end{proof}

\begin{proof}{Proof of Theorem~\ref{T:ground}}
First, we claim that
\begin{equation}
\label{E:minH_N}
\min_{\bS}\mathcal H_N(\bS)=-\max\{J_1,J_2\}|\T_N|.
\end{equation}
The inequality $\le$ is seen by taking $\bS_{\br}\equiv\hate_1$ or $\bS_{\br}\equiv\hate_2$, depending on whether $J_1\ge J_2$ or not, and evaluating~$\mathcal H_N(\bS)$ for this choice. For the opposite bound we use   $J_1,J_2>0$ to drop  the first two terms in \eqref{E:H-rewrite} and conclude
\begin{equation}
\min_{\bS}\mathcal H_N(\bS)\ge -|\T_N|
\max_{S\in O(2)}\bigl(J_1[S^x]^2+J_2[S^z]^2\bigr)\ge-\max\{J_1,J_2\}|\T_N|.
\end{equation}
This also shows that the minimum is attained only by configurations for which
\begin{equation}
S_{\br}^{\alpha}=S_{\br'}^{\alpha} \quad\mbox{for the endpoints $\br,\br'$ of all $\alpha$-bonds $\langle\br,\br'\rangle_\alpha$},\,\alpha=x,z,
\end{equation}
and
\begin{equation}
\label{E:cond1}
J_1[S_{\br}^x]^2+J_2[S_{\br}^z]^2 = \max\{J_1,J_2\}\quad \mbox{for all }\br\in\T_N
\end{equation}
hold true.

Let now $\bS$ be a ground state of $\mathcal H_N$. The above findings guarantee that~$S_{\br}^x$ is constant on any $x$-plaquette and~$S_{\br}^z$ is constant on any $z$-plaquette. Applications of $\bvarphi_{\br}$ to the plaquettes where the corresponding value is negative defines a configuration $\tilde{\bS}$ where $\tilde S_{\br}^{\alpha}\ge0$ for $\alpha=x,z$ and all $\br\in\T_N$. (This is the desired modification of $\bS$ by applications of the maps $(\bvarphi_{\br})$.) Let now $\br$ be a lower-left vertex of a $z$-plaquette and let $\br':=\br-\hate_1$ be its neighbor to the left. Since $\tilde{\bS}$ is also a ground state and the edge $\langle\br,\br'\rangle$ is an $x$-edge, we have $\tilde S_{\br}^x=\tilde S_{\br'}^x$. However, the components $\tilde S_{\br}^z,\tilde S_{\br'}^z$ are both non-negative and since one component of a spin determines the other up to a sign, we also have $\tilde S_{\br}^z=\tilde S_{\br}^z$. It follows that $\tilde{\bS}_{\br}=\tilde{\bS}_{\br'}$. Proceeding similarly for all pairs of neighbors in~$\T_N$ we conclude that $\tilde{\bS}_{\br}=\hate$ for some unit vector $\hate\in\{ \bv \in\mathbb R^2\colon  \bv\cdot\be_i\ge 0, i=1,2\}$ and all $\br\in\T_N$.

It remains to determine the set of vectors $\hate$ that are admissible at given values of the parameters of the model. We have
\begin{equation}
\mathcal H_N(\tilde{\bS})=-\bigl(J_1(\hate\cdot\hate_1)^2+J_2(\hate\cdot\hate_2)^2\bigr)|\T_N|.
\end{equation}
Thus, when $J_1>J_2$, we must have $\hate=\hate_1$ while $\hate=\hate_2$ when $J_2>J_1$. At the symmetry point, $J_1=J_2$, any $\hate$ will give the same value. We have thus shown that $\bS$ is a modification of a constant configuration (namely~$\tilde{\bS}$) of the desired type in all cases of interest.
\end{proof}

\section{Technical ingredients}
\label{sec-IV}

In this section we assemble the technical ingredients needed for the proof of the main theorems concerning the phase transition in the model of interest. The proofs will come in Section~\ref{sec-V}.

\subsection{Chessboard estimates}
The proof of the positive-temperature part of the results will be based on the technique of \emph{chessboard estimates}, based on \emph{reflection positivity}, whose origins go to the seminal work of Dyson, Fr\"ohlich, Israel, Lieb, Simon and Spencer from the late 1970s. This technique, along with a related \emph{infrared-bound} technology, has proved extremely useful in establishing symmetry-breaking phase transitions in various classical and quantum 
systems with a continuous symmetry~\cite{Frohlich:1976p1809,Dyson:1976p5123,Frohlich:1978p1239,Frohlich:1980p5126}, 
order-disorder transitions~in the Potts and related models~\cite{Kotecky:1982p124,Chayes:1995p107,Biskup:2000p113}, 
low-temperature ordering in liquid-crystal models~\cite{Heilmann:1979p5131,Zagrebnov:1996p5133}. More recently, this technique has also been used to prove phase transitions in systems with highly degenerate ground states 
without an underlying symmetry~\cite{Biskup:2005p701,Biskup:2004p697} including gradient fields with a non-convex interaction~\cite{Biskup:2007p109}. The theoretical foundations of this technique are well developed already  
in the original papers~\cite{Frohlich:1978p1239,Frohlich:1980p5126}; the more recent developments are summarized in the lecture notes~\cite{Biskup-book}.

Consider the model with the Hamiltonian $\mathcal H_N$ on the torus $\T_N$ with $N$ even. Consider an even integer $B$ that divides $N$ and let $\Lambda_B:=\{0,1,\dots,B-1\}\times\{0,1,\dots,B-1\}$ denote the block of $B\times B$ vertices with the lower-left corner at the origin. Consider a regular partitioning of $\T_N$ into pairwise disjoint translates of $\Lambda_B$ by vectors from $\{B\br\colon\br\in\T_{N/B}\}$. For $\br\in\T_{N/B}$, let $\vartheta_{\br}$ denote the translation by $B\br$. On the configuration space,
\begin{equation}
\bigl[\vartheta_{\br}(\bS)\bigr]_{\br'}=\bS_{\br'-B\br},\qquad \br'\in\T_N,\,\br\in\T_{N/B}.
\end{equation}
We call an event $\mathcal A$ a \emph{$B$-block event} if~$\mathcal A$ depends only on $\{S_{\br}\colon\br\in\Lambda_B\}$ and we use $\mathscr F_B$ to denote the collection of all $B$-block events.

For each $\mathcal A\in\mathscr F_B$ we now define a family $\{\theta_{\br}(\mathcal A)\colon\br\in\T_{N/B}\}$ of translations-reflections as follows. First, let $\mathcal A_1$ denote the reflection of $\mathcal A$ through the horizontal mid-line $\{(\frac{B-1}2,y)\colon y\in\R/(N\R)\}$ halving the box~$\Lambda_B$. Similarly, we use $\mathcal A_2$ to denote the reflection of $\mathcal A$ through the vertical mid-line $\{(x,\frac{B-1}2)\colon x\in\R/(N\R)\}$ and $\mathcal A_{12}$ to denote the reflection of $\mathcal A$ through both lines (the two reflections commute and so the order in which they are taken is immaterial). Note that $\mathcal A_1,\mathcal A_2,\mathcal A_{12}\in\mathscr F_B$. For $\br = (r_1,r_2)\in\T_{N/B}$, we then set
\begin{equation}
\theta_{\br}(\mathcal A):=\cases{
\vartheta_{\br}^{-1}(\mathcal A),&\qquad for $r_1,r_2$ even,\\
\vartheta_{\br}^{-1}(\mathcal A_1),&\qquad for $r_1$ odd and $r_2$ even,\\
\vartheta_{\br}^{-1}(\mathcal A_2),&\qquad for $r_1$ even and $r_2$ odd,\\
\vartheta_{\br}^{-1}(\mathcal A_{12}),&\qquad for $r_1,r_2$ odd.\\}
\end{equation}
Here $\vartheta^{-1}(\mathcal A):=\{\vartheta_{\br}(\bS)\colon\bS\in\mathcal A\}$. Notice that $\vartheta_{\br}(\mathcal A)$ thus depends only on the part of the spin configuration in the block $B\br+\Lambda_B$. We then have:

\begin{lemma}[Chessboard estimate]
\label{T:chessboard}
Suppose $J_1,J_2\ge0$ 
and $\beta\ge0$.
Then for any events $\mathcal A_1,\dots,\mathcal A_m\in\mathscr F_B$ and any distinct $\br_1,\dots,\br_m\in\T_{N/B}$,
\begin{equation}
\mu_{N,\beta}\biggl(\,\bigcap_{i=1}^m\theta_{\br_i}(\mathcal A_i)\biggr)
\le\prod_{i=1}^m\biggl[\,\mu_{N,\beta}\biggl(\,\bigcap_{\br\in\T_{N/B}}\theta_{\br}(\mathcal A_i)\biggr)\biggr]^{(B/N)^2}.
\end{equation}
\end{lemma}

The punchline of this result is that the probability of a simultaneous occurrence of several (often undesirable) events on the torus is bounded by the product of the probabilities of events where the individual $B$-block events have been \emph{disseminated}  ---  using the maps $\theta_{\br}$  ---  throughout the entire torus. The latter quantities are often rather explicitly computable.

\begin{proof}{Proof of Lemma~\ref{T:chessboard}}
As already alluded to, the key input for the chessboard estimates is reflection positivity of the interaction. We will now define the necessary concept and check the validity of this property. Let $P$ be a plane bisecting a horizontal or vertical line of edges of $\T_N$, i.e., $P$ is of the form either
\begin{equation}
\bigl\{(n+\ffrac12,y)\colon y\in\R/(N\R)\bigr\}\cup\bigl\{(n+\ffrac N2+\ffrac12,y)\colon y\in\R/(N\R)\bigr\},
\end{equation}
with $n=0,1,\dots,\frac N2-1$, or
\begin{equation}
\bigl\{(x,n+\ffrac12)\colon x\in\R/(N\R)\bigr\}\cup\bigl\{(x,n+\ffrac N2+\ffrac12)\colon x\in\R/(N\R)\bigr\},
\end{equation}
with $n=0,1,\dots,\frac N2-1$. The plane has two components and it splits the torus into a left half $\T_N^-$ and the right half $\T_N^+$. Abusing the notation slightly, let~$\theta_P$ denote the map on the configuration space representing the reflection $\T_N^+\leftrightarrow\T_N^-$.

A sufficient condition for the interaction to be reflection positive is that for each such a plane~$P$, there are functions $g$ and $h=(h_i)$ depending only on $\{\bS_{\br}\colon\br\in\T_N^+\}$ so that
\begin{equation}
-\mathcal H_N=g+\theta_P(g)+\sum_i h_i\theta_P(h_i).
\end{equation}
We will now demonstrate that $\mathcal H_N$ is indeed of this form. Let
\begin{equation}
g(\bS):=-J_1\sum_{\vbox{
\hbox{\,\,$\scriptscriptstyle\langle \br,\br'\rangle_x$}
\vglue-6pt
\hbox{$\scriptscriptstyle\br,\br'\in\T_N^+$}}}
\, S_{\br}^{x}S_{\br'}^{x}-J_2\sum_{\vbox{
\hbox{\,\,$\scriptscriptstyle\langle \br,\br'\rangle_z$}
\vglue-6pt
\hbox{$\scriptscriptstyle\br,\br'\in\T_N^+$}}}
 \, S_{\br}^{z}S_{\br'}^{z} 
\hfill
\end{equation}
and note that $g$ depends only on the spins in $\T_N^+$. The collection of functions $h$ will be parametrized by the vertices $\br\in\T_N^+$ that have an edge to a vertex $\br'\in\T_N^-$. We will use~$P^+$ to denote the set of such vertices $\br$. We set
\begin{equation}
h_{\br}:=\cases{\sqrt{J_1}\,S_{\br}^x,&if $\langle\br,\br'\rangle$ is an $x$-edge,\\
\sqrt{J_2}\,S_{\br}^z,&if $\langle\br,\br'\rangle$ is a $z$-edge,}
\qquad \br\in P^+,
\end{equation}
where~$\br'$ stands for the reflection of~$\br$ through plane~$P$.
Then $h_{\br}\theta_P(h_{\br})=J_\alpha S_{\br}^{\alpha}S_{\br'}^{\alpha}$, with $\alpha$ depending on the type of the edge $\langle\br,\br'\rangle$. A moment's thought then shows that
\begin{equation}
-\mathcal H_N=g+\theta_P(g)+\sum_{\br\in P^+} h_{\br}\,\theta_P(h_{\br})
\end{equation}
and so the interaction is of the desired form. As the \emph{a priori} measure on the spins has a product structure, standard theory (cf~\cite[Theorem~4.1]{Frohlich:1978p1239} or \cite[Theorem~5.8]{Biskup-book}) readily implies the desired claim.
\end{proof}

\subsection{Gaussian calculations}
Through the use of chessboard estimates, the proof of the phase transition will be reduced to some tedious but explicit computations of multivariable Gaussian integrals. Informally, these can be understood as calculations of spin-wave free energies corresponding to the spin system at hand. The goal of this section is to carry out these calculations and derive the necessary estimates between actual partition functions and their Gaussian approximations. Throughout we will assume that
\begin{equation}
J_1=J_2=:J
\end{equation}
with $J>0$.

For a given unit vector $\hate(\theta):=(\cos\theta,\sin\theta)\in\R^2$ and $\Delta>0$, we define the quantity
\begin{equation}
\label{E:Z-integral}
\mathcal Z_N(\theta,\Delta):=\texte^{-\beta J|\T_N|}\int\texte^{-\beta\mathcal H_N(\bS)}\prod_{\br\in\T_N}1_{\{|\bS_{\br}-\hate(\theta)|<\Delta\}}\prod_{\br\in\T_N}
\nu(\textd\bS_{\br}).
\end{equation}
This is the partition function restricted to configurations within $\Delta$ of a constant configuration pointing in direction of the unit vector $\hate(\theta)$. Next, we define a function $\theta\mapsto F(\theta)$ as follows. For each $\bk:=(k_1,k_2)\in[-\pi,\pi]^2$, let us introduce the quantities
\begin{equation}
\label{E:3.17}
a_\pm:=1\pm\texte^{-\texti k_1},\quad b_\pm:=1\pm\texte^{-\texti k_2},\quad\mbox{and}\quad\rho:=-\cos(2\theta).
\end{equation}
Consider the matrix 
\begin{equation}
\label{E:Pi}
\!\!\!\!\!\!\!\!\!\!\!\!\!\!\!\!\!\!\!\!\!\!\!\!\!\!
M(\bk,\theta):=\frac12\left(\,
\matrix{
|a_-|^2+|b_-|^2 & \rho a_-a_+^\ast & \rho b_-b_+^\ast & 0
\cr 
\rho a_-^\ast a_+ & |a_+|^2+|b_-|^2 & 0 & \rho b_-b_+^\ast 
\cr
\rho b_-^\ast b_+  & 0 & |a_-|^2+|b_+|^2 & \rho a_-a_+^\ast
\cr
0 & \rho b_-^\ast b_+ & \rho a_-^\ast a_+  & |a_+|^2+|b_+|^2
}
\,\right).
\end{equation} 
We will  see in a moment that $\det M(\bk,\theta)\ge0$ and so we may define
\begin{equation}
\label{E:3.19}
F(\theta):=\frac12\log(\beta J)+\frac18\int_{[-\pi,\pi]^2}\frac{\textd\bk}{(2\pi)^2}\log\det M(\bk,\theta).
\end{equation}
The key facts about the quantity $F(\theta)$ and its relation to $\mathcal Z_N(\theta,\Delta)$ are the subject of the following two claims:

\begin{prop}
\label{P:propertiesofM}
For any $\bk:=(k_1,k_2)\in[-\pi,\pi]^2$ and any $\theta\in[0,2\pi]$ we have
\begin{equation}
\label{E:det-M-bd}
\sin^2(2\theta)\sin^2(k_1) \sin^2(k_2)\le\det M(\bk,\theta)\le16\sin^2(2\theta).
\end{equation}
In particular, $F(\theta)$ is finite for all $\theta$ with $\sin(2\theta)\ne0$. 
The function $\theta\mapsto F(\theta)$ is periodic with period $\ffrac\pi2$, symmetric and continuous on the interval $(0,\ffrac\pi2)$, increasing on $(0,\ffrac\pi4)$ and decreasing on $(\ffrac\pi4,\ffrac\pi2)$.
The infimum of $F$ is $-\infty$ and it is achieved exactly at $\theta\in\{0,\ffrac\pi2,\pi,\ffrac{3\pi}2\}$. See Fig.~2.
\end{prop}

\begin{figure}[ht]
\label{fig2}
\begin{center}
\epsfxsize=4.9in
\epsfbox{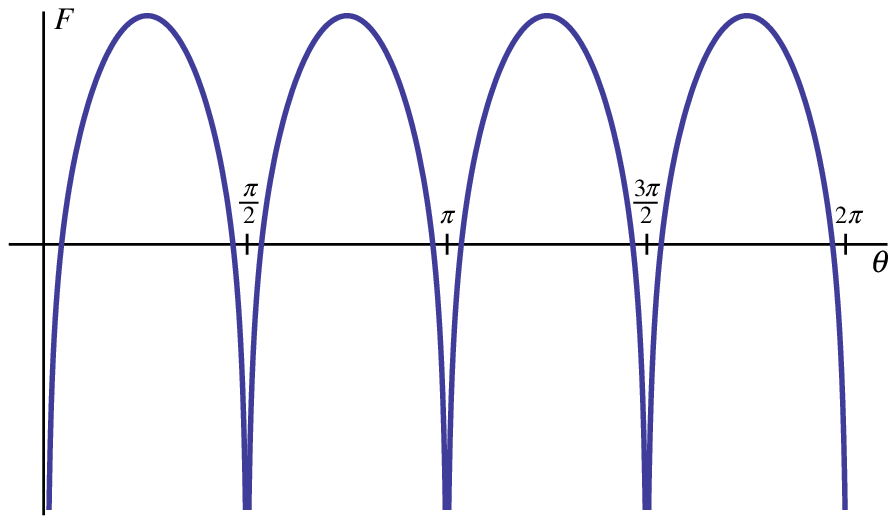} 
\end{center}
\noindent
\begin{minipage}{0.85\textwidth}
\small
\caption{The plot of function $\theta\mapsto F(\theta)$ for $\theta$ ranging from 0 to $2\pi$. The key fact is that~$F$ tends to minus infinity as $\theta$ tends to multiples of $\ffrac\pi2$. However, note that~$F$ is relevant for the approximation of $N^{-2}\log\mathcal Z_N(\theta,\Delta)$ to within~$\tau$ only when~$\theta$ avoids the region where~$|\sin(2\theta)|\le\tau$ --- i.e., exactly the vicinity of its singularity points.}
\end{minipage}
\end{figure}

\begin{prop}
\label{P:Gaussian}
Suppose that $J_1=J_2=:J>0$. For each $\tau>0$ there are numbers $\delta>0$ and $N_0<\infty$  such that if 
\begin{equation}
\label{Delta-cond}
\beta J\Delta^2>\frac1\delta,\quad\beta J\Delta^3<\delta,
\end{equation}
then 
\begin{equation}
\label{E:freeenapprox}
\Bigl|\frac1{N^2}\log\mathcal Z_N(\theta,\Delta)+F(\theta)\Bigr|<\tau
\end{equation}
holds true whenever $N\ge N_0$ and  $|\sin(2\theta)|>\tau$.
\end{prop}

We will first focus on Proposition~\ref{P:Gaussian} because its proof explains the reasons underlying the introduction of the quantity $F$.
The proof consists of a sequence of approximations followed by a standard diagonalization of a multivariate Gaussian integral. Let us write each $\bS_{\br}$ as
\begin{equation}
\label{bS-btheta}
\bS_{\br}=\bigl(\cos(\theta+\vartheta_{\br}),\sin(\theta+\vartheta_{\br})\bigr).
\end{equation}
On the event entering the integral \eqref{E:Z-integral}, $\vartheta_{\br}=O(\Delta)$, so if $\Delta$ is small (which we may assume since \eqref{Delta-cond} forces $\Delta<\delta^2$), then $\bvartheta:=(\vartheta_{\br})$ may be regarded as deviation variables. The conversion to a Gaussian integral is performed as follows:

\begin{lemma}
Suppose~$\Delta<1$ and define the function
\begin{equation}
G_{N,\theta}(\bvartheta):=\frac 12\sum_{\langle\br,\br'\rangle_x}(\vartheta_{\br}-\vartheta_{\br'})^2\sin^2(\theta)
+\frac 12\sum_{\langle\br,\br'\rangle_z}(\vartheta_{\br}-\vartheta_{\br'})^2\cos^2(\theta).
\end{equation}
There exists a constant $c\in(0,\infty)$ such that for every $\bS$ that is related to $\bvartheta$ via \eqref{bS-btheta}, with $\vartheta_{\br}\in(-\pi,\pi]$, and obeys $|\bS_{\br}-\hate(\theta)|<\Delta$ at all $\br\in\T_N$,
\begin{equation}
\bigl|\mathcal H_N(\bS)+J|\T_N|-JG_{N,\theta}(\bvartheta)\bigr|\le cJ|\T_N|\Delta^3.
\end{equation}
\end{lemma}

\begin{proof}{Proof}
With the restriction $\vartheta_{\br}\in(-\pi,\pi]$ and~$\Delta<1$, the correspondence between~$\bS$ and~$\bvartheta$ is one-to-one. The restriction $|\bS_{\br}-\hate(\theta)|<\Delta$ implies~$\vartheta_{\br}=O(\Delta)$ uniformly in~$\br\in\T_N$.
The claim now follows by writing $\mathcal H_N$ in the form \eqref{E:H-rewrite}, and noting that
\begin{equation}
S_{\br}^x-S_{\br'}^x=(\vartheta_{\br}-\vartheta_{\br'})\sin\theta+O(\Delta^2)
\end{equation}
and similarly for $S_{\br}^z-S_{\br'}^z$.
\end{proof}

If we set $\Delta':=2\arcsin(\ffrac\Delta2)$ for the maximal angle between $\bS_{\br}$ and $\hate(\theta)$ allowed by the constraints, we thus have
\begin{equation}
\label{E:3.51}
\mathcal Z_N(\theta,\Delta)=\texte^{O(\Delta^3)\beta J|\T_N|}\int_{\R^{\T_N}}\texte^{-\beta J G_{N,\theta}(\vartheta)}\chi_{\Delta',N}(\bvartheta)\prod_{\br\in\T_N}\frac{\textd\vartheta_{\br}}{\sqrt{2\pi}},
\end{equation}
where $\sqrt{2\pi}$ appears due to our normalization convention for $\nu$ and where
\begin{equation}
\chi_{\Delta',N}(\bvartheta):=\prod_{\br\in\T_N}1_{(-\Delta',\Delta')}(\vartheta_{\br}).
\end{equation}
In order to evaluate the expression in \eqref{E:3.51}, we will notice the following relation between the function $G_{N,\theta}$ and the matrices $M(\bk,\theta)$. Consider the Fourier-reciprocal torus, 
\begin{equation}
\T_N^\star:=\bigl\{{\textstyle\frac{2\pi}N}(n_1,n_2)\colon n_1,n_2=0,\dots,N-1\bigr\},
\end{equation}
and let $(\T_N^\star)_+$ be its first quadrant, i.e., the collection of $\bk=(k_1,k_2)\in\T_N^\star$ satisfying $0\le k_1,k_2<\pi$. Let $M(\theta)_{\bk,\bk'}$ be the $N^2\times N^2$ matrix indexed by $\bk,\bk'\in\T_N^\star$ that is block diagonal and whose entry is zero unless $\bk,\bk'\in\{\bar{\bk},\bar{\bk}+\pi\hate_1,\bar{\bk}+\pi\hate_2,\bar{\bk}+\pi\hate_1+\pi\hate_2)$ for some $\bar{\bk}\in(\T_N^\star)_+$, and whose entries for these four values are collected (in the given order) in the $4\times4$-matrix $M(\bk,\theta)$ defined above.

\begin{lemma}
\label{lemma3.6}
We have
\begin{equation}
\label{E:GN-M}
G_{N,\theta}(\bvartheta)=\sum_{\br,\br'\in\T_N}\widehat M(\theta)_{\br,\br'}\vartheta_{\br}\vartheta_{\br'},
\end{equation}
where
\begin{equation}
\label{E:hatM}
\widehat M(\theta)_{\br,\br'}:=\frac1{|\T_N|}\sum_{\bk,\bk'\in\T_N^\star}M(\theta)_{\bk,\bk'}\texte^{\texti(\bk-\bk')\cdot\br}.
\end{equation}
\end{lemma}

\begin{proof}{Proof}
We may write $G_{N,\theta}$ as
\begin{equation}
G_{N,\theta}(\bvartheta)=\sum_{\br\in\T_N}\sum_{\hate=\hate_1,\hate_2} J_{\br,\hate}(\vartheta_{\br}-\vartheta_{\br+\hate})^2,
\end{equation}
where $J_{\br,\hate_1}:=\sin^2(\theta)$ for vertices $\br$ with even first coordinate and $J_{\br,\hate_1}:=\cos^2(\theta)$ for vertices~$\br$ with odd first coordinate, and same for $J_{\br,\hate_2}$ and the second coordinate of $\br$. Invoking the (discrete) Fourier representation $\vartheta_{\br}=|\T_N|^{-1/2}\sum_{\bk\in\T_N^\star}\widehat\vartheta_{\bk}\,\texte^{-\texti\bk\cdot\br}$ we now get\begin{equation}
\label{E:G_Nrepr}
G_{N,\theta}(\bvartheta)=\sum_{\bk,\bk'\in\T_N}\sum_{\hate=\hate_1,\hate_2} \widehat\vartheta_{\bk}\widehat\vartheta_{\bk'}^\ast(1-\texte^{-\texti\bk\cdot\hate})(1-\texte^{\texti\bk'\cdot\hate})\hat J_{\bk,\bk'}(\hate),
\end{equation}
where
\begin{equation}
\hat J_{\bk,\bk'}(\hate):=\frac1{|\T_N|}\sum_{\br\in\T_N}J_{\br,\hate}\,\texte^{\texti(\bk'-\bk)\cdot\br}.
\end{equation}
Now, since the couplings $J_{\br,\hate}$ are 2-periodic in direction $\hate$ and translation invariant in the complementary direction, $\hat J_{\bk,\bk'}(\hate)$ will be non-zero only when $\bk'=\bk$ or $\bk'=\bk+\pi\hate$. In these two cases we get
\begin{equation}
\hat J_{\bk,\bk}(\hate)=\frac12\quad\mbox{and}\quad\hat J_{\bk,\bk+\pi\hate}(\hate)=-\frac12\cos(2\theta).
\end{equation}
Consider a fixed vector $\bar{\bk}\in(\T_N^\star)_+$ and consider the restriction of the sum in \eqref{E:G_Nrepr} to $\bk,\bk'\in\{\bar{\bk},\bar{\bk}+\pi\hate_1,\bar{\bk}+\pi\hate_2,\bar{\bk}+\pi\hate_1+\pi\hate_2\}$. Let $a_\pm,b_\pm$ denote the quantities in \eqref{E:3.17} for this $\bar{\bk}$. The quadratic form corresponding to $\hate:=\hate_1$ is then described by the matrix
\begin{equation}
\label{E:Pi1}
\frac12\left(\,
\matrix{
|a_-|^2 & \rho a_-a_+^\ast & 0 & 0
\cr 
\rho a_-^\ast a_+ & |a_+|^2 & 0 & 0 
\cr
0  & 0 & |a_-|^2 & \rho a_-a_+^\ast
\cr
0 & 0 & \rho a_-^\ast a_+  & |a_+|^2
}
\,\right),
\end{equation} 
while the contribution corresponding to $\hate:=\hate_2$ is described by the matrix
\begin{equation}
\label{E:Pi2}
\frac12\left(\,
\matrix{
|b_-|^2 & 0 & \rho b_-b_+^\ast & 0
\cr 
0 & |b_-|^2 & 0 & \rho b_-b_+^\ast 
\cr
\rho b_-^\ast b_+  & 0 & |b_+|^2 & 0
\cr
0 & \rho b_-^\ast b_+ & 0  & |b_+|^2
}
\,\right).
\end{equation} 
Adding these contributions together, we obtain \eqref{E:GN-M}--\eqref{E:hatM} with \eqref{E:Pi}.
\end{proof}

\begin{proof}{Proof of Proposition~\ref{P:Gaussian}}
Let $I(\theta,\Delta)$ denote the integral  in \eqref{E:3.51}.
Our goal is to evaluate $I(\theta,\Delta)$ to within  multiplicative correction of the order $\texte^{O({\tau})|\T_N|}$. If it were not for the indicator~$\chi_{\Delta',N}(\bvartheta)$, the integral would be Gaussian; unfortunately, as $G_{N,\theta}(\bvartheta)$ depends only on the differences $\vartheta_{\br}-\vartheta_{\br'}$, it would also diverge. We will therefore have to treat the indicator with some extra care by deriving suitable upper and lower bounds. 

First, for any $\lambda>0$,
\begin{equation}
\chi_{\Delta',N}(\bvartheta)\le\exp\Bigl\{\frac\lambda2\beta Jc'\Delta^2|\T_N|-\frac\lambda2\beta J\sum_{\br\in\T_N}\vartheta_{\br}^2\Bigr\},
\end{equation}
where $c':=\sup_{|\Delta|<2}(\Delta'/\Delta)^2\in[1,\infty)$. Substituting this bound into $I(\theta,\Delta)$, we now scale $\vartheta_{\br}$ by~$\sqrt{\beta J}$ and use Lemma~\ref{lemma3.6} to diagonalize the quadratic form in the exponent to get
\begin{equation}
I(\theta,\Delta)\le(\beta J)^{-\frac12|\T_N|}\,\texte^{\frac12\lambda\beta Jc'\Delta^2|\T_N|}\prod_{\bk\in(\T_N^\star)_+}\frac1{\sqrt{\det(\lambda+M(\bk,\theta))}}.
\end{equation}
The $\lambda$ in the denominator now regularizes the contribution of the $\bk=0$ mode. We thus get the bound
\begin{equation}
\label{E:3.54}
\frac{\log\mathcal Z_N(\theta,\Delta)}{|\T_N|}\le c\beta J\Delta^3+\frac\lambda2\beta Jc'\Delta^2-F_N(\theta,\lambda),
\end{equation}
where
\begin{equation}
F_N(\theta,\lambda):=\frac12\log(\beta J)+\frac12\frac1{|\T_N|}\sum_{\bk\in(\T_N^\star)_+}\log\det\bigl[\lambda+ M(\bk,\theta)\bigr].
\end{equation}
Since $F_N(\theta,\lambda)$ is, in the limit $N\to\infty$, finite and continuous in $\theta$ and also obviously larger than $F(\theta)$ for any $\lambda>0$, it suffices to choose~$\lambda$ so that the first two terms on the right of \eqref{E:3.54} are strictly less than~$\tau$.
In particular, taking $\lambda:=\Delta$, it suffices to choose $\delta$ so that $(c+\frac12c')\delta\le \tau$
and pick $N_0$ so that $F_N(\theta,\Delta)>F(\theta)$ for $N\ge N_0$ and all $\theta$, to get $N^{-2}\log\mathcal Z_N(\theta,\Delta) + F(\theta)\le \tau$ whenever $\beta J\Delta^3\le \delta$.

The requisite lower bound is derived by a change-of-measure argument. Consider the Gaussian measure underlying the upper bound above:
\begin{equation}
\P_\lambda(\textd\bvartheta):=\texte^{|\T_N|F_N(\theta,\lambda)-\beta J G_{N,\theta}(\bvartheta)}\exp\Bigl\{-\frac\lambda2\beta J\sum_{\br\in\T_N}\vartheta_{\br}^2\Bigr\}\prod_{\br\in\T_N}\frac{\textd\vartheta_{\br}}{\sqrt{2\pi}}.
\end{equation}
Using $\E_\lambda$ to denote the corresponding expectation, for any $\lambda>0$ we clearly have
\begin{equation}
\label{E:3.57}
\mathcal Z_N(\theta,\Delta)\ge\texte^{-c\beta J\Delta^3|\T_N|-|\T_N|F_N(\theta,\lambda)}\E_\lambda(\chi_{\Delta',N}).
\end{equation}
Since $F_N(\theta,\lambda)$ will tend to $F(\theta)$ when $N\to\infty$ and $\lambda\downarrow0$, uniformly in $N$ with $|\sin(\theta)|>\tau$, it will suffice to show that, for any $\lambda>0$ small enough, $\E_\lambda(\chi_{\Delta',N})^{1/|\T_N|}$ is near one uniformly as~$N\to\infty$. To this end, we first use chessboard estimates  ---  exactly in the setting described above with $B:=1$  ---  for the Gaussian measure $\mathbb P_\lambda$ to derive
\begin{equation}
\E_\lambda(\chi_{\Delta',N})=\E_\lambda\Bigl(\,\prod_{\br\in\T_N}1_{\{|\vartheta_{\br}|<\Delta'\}}\Bigr)
\ge\prod_{\br\in\T_N}\P_\lambda\bigl(|\vartheta_{\br}|<\Delta'\bigr).
\end{equation}
Next the Chebyshev inequality yields
\begin{equation}
\P_\lambda\bigl(|\vartheta_{\br}|\ge\Delta'\bigr)\le\P_\lambda\bigl(|\vartheta_{\br}|\ge\Delta\bigr)\le
\frac{\E_\lambda(\vartheta_{\br}^2)}{\Delta^2}=\frac{\mbox{Var}_\lambda(\vartheta_{\br})}{\Delta^2},
\end{equation}
where we used that~$\Delta'\ge\Delta$.
The variance of a Gaussian variable increases when we make the corresponding quadratic form in the exponent smaller (as a matrix). As this quadratic form is bounded below by $\lambda\beta J$ times identity, the thus get
\begin{equation}
\label{GLB}
\P_\lambda\bigl(|\vartheta_{\br}-\theta|<\Delta\bigr)\ge1-
\frac1{\lambda\beta J\Delta^2}.
\end{equation}
To derive the matching lower bound from \eqref{E:3.57}, we thus need to choose $\lambda$ small so that $F_N(\theta,\lambda)$ is already close to $F(\theta)$ but such that $\lambda\beta J\Delta^2$ is still large.
A specific choice that will work is as follows: Pick $\delta>0$ so small that, for some $N_0$, we have $\abs{F_N(\theta,\sqrt\delta)-F(\theta)}<\frac{\tau}2$
for all $N\ge N_0$ and all $\theta$ with $|\sin(2\theta)|>\tau$. In addition, assume that also $c\delta-\log(1-\sqrt\delta)<\frac\tau2$. The choice $\lambda:=\sqrt\delta$ and the bounds \eqref{E:3.57} and \eqref{GLB} then yield the desired lower bound on $\log\mathcal Z_N(\theta,\Delta)$ once both conditions \eqref{Delta-cond} are satisfied.
\end{proof}

\begin{proof}{Proof of Proposition~\ref{P:propertiesofM}}
Assume that $\bk$ is such that $a_\pm,b_\pm\ne0$ and recall the definition of~$\rho$.
First, notice that $\det M(\bk,\theta)=0$ when $\rho=\pm1$. The case $\rho=1$ is easily checked by multiplying the matrix $M(\bk,\theta)$ by the vector $(1,-1,-1,1)$; the case $\rho=-1$ is checked using the vector $(1,1,1,1)$. Since  $\det M(\bk,\theta)$ is an even quartic polynomial in $\rho$ that is divisible by $(1-\rho^2)$, it can be written in the form
\begin{equation}
\label{E:det-tilde-M}
\det M(\bk,\theta)=(1-\rho^2)(A-\rho^2 C).
\end{equation}
An explicit computation yields
\begin{equation}
A:=\frac1{16}\bigl(|a_-|^2+|b_-|^2\bigr)\bigl(|a_+|^2+|b_-|^2\bigr)\bigl(|a_-|^2+|b_+|^2\bigr)\bigl(|a_+|^2+|b_+|^2\bigr)
\end{equation}
and
\begin{equation}
C:=\frac1{16}\bigl(|a_-|^2|a_+|^2-|b_-|^2|b_+|^2\bigr)^2.
\end{equation}
Obviously, under the conditions $a_\pm,b_\pm\ne0$ we have $A>0$ and $C\ge0$. Moreover, since $|a_\pm|,|b_\pm|\le2$, we have $A\le 16$. Therefore
\begin{equation}
\frac{|a_-|^2|a_+|^2|b_-|^2|b_+|^2}{16}\le A-C\le A-\rho^2C\le A\le16.
\end{equation}
The left-hand side equals $\sin^2(k_1)\sin^2(k_2)$ and so \eqref{E:det-M-bd} follows by plugging the above inequalities  into \eqref{E:det-tilde-M}.

To get the second part of the claim, let $g(x):=(1-x)(A-Cx)$ be the function appearing on the right-hand side of \eqref{E:det-tilde-M}. Then $g'(0)=-A-C<0$ and $g'(1)=-A+C<0$. But $g$ is quadratic and so it is strictly decreasing throughout $[0,1)$. In particular, $g(1)<g(x)$ for all $x\in[0,1)$. Using this in \eqref{E:3.19}, the desired claims follow.
\end{proof}

\subsection{Good and bad events}
For the proof of our key estimates, we will follow, as in the previous sections, the general scheme  developed in \cite{Biskup:2005p701,Biskup:2004p697} and further discussed in~\cite[Sect.~6.4]{Biskup-book}. For  a positive constant $\eta$  and an even integer~$B$ that divides~$N$, define the \emph{good} $B$-block events $\Gcal_x$ and $\Gcal_z$ as follows. First set
\begin{equation}
\Gcal_x^0:=\bigcap_{\br\in\Lambda_B}\bigl\{\bS\colon\,\,|\bS_{\br}-\be_1|<\eta\bigr\},
\end{equation}
and similarly for~$\Gcal_z^0$ with~$\be_1$ replaced by~$\be_2$.
Now let $\Lambda_B^0$ denote the even-sublattice vertices in $\Lambda_B$; these are the lower-left corners of either $x$ or~$z$-plaquettes.  We use $\Lambda_B^1$ to denote the set of those vertices in $\Lambda_B^0$
that are lower-left corners of $x$-plaquettes, $\Lambda_B^2=\Lambda_B^0\setminus \Lambda_B^1$.
For any set $\Lambda\subset\Lambda_B^0$, let $\bvarphi_\Lambda(\bS)$ be the composition of $\bvarphi_{\br}$ for all $\br\in\Lambda$. As these maps commute, the order of the composition is immaterial. With the help of these notations, we now set
\begin{equation}
\Gcal_\alpha:=\bigcup_{\Lambda\subset\Lambda_B^0}\bvarphi_\Lambda(\Gcal_\alpha^0),\qquad\alpha=1,2.
\end{equation}
These are the \emph{good} $B$-block events; the requisite \emph{bad} event is defined by
\begin{equation}
\Bcal:=(\Gcal_x\cup\Gcal_z)^{\rm c}.
\end{equation}
For a $B$-block event $\Acal$ define the quantity
\begin{equation}
\label{E:pbeta}
\mathfrak p_\beta(\Acal):=\lim_{N\to\infty}\biggl[\,\mu_{N,\beta}\biggl(\,\,\bigcap_{\br\in\T_{N/B}}\theta_{\br}(\mathcal A)\biggr)\biggr]^{(B/N)^2},
\end{equation}
where $N$ is taken to infinity along the even multiples of~$B$. The limit exists by standard subadditivity arguments. Note that this quantity is a limiting version of the objects on the right-hand side of the chessboard estimates. A key input for our proofs is the observation that $\mathfrak p_\beta(\Bcal)$ is small:

\begin{prop}
\label{P:smallBcal}
For each $\eta>0$ and ${\e}>0$ there exists $\beta_0<\infty$ such that for any $\beta\ge \beta_0$
there exists $B$ for which the bad event $\Bcal$ defined using these $\eta$ and $B$ satisfies
\begin{equation}
\mathfrak p_\beta(\Bcal)<{\e}.
\end{equation}
\end{prop}

In order to show that the bad event is unlikely to occur, we will need to further partition it into several subevents. First, consider a number $\Delta>0$  ---  to be evetually chosen in dependence of $\beta$ so that 
the conditions \eqref{Delta-cond} hold true  ---  and let us isolate the configurations where the corresponding component of two neighbouring spins somewhere in~$\Lambda_B$ differ by more than a quantity proportional to~$\Delta$:
\begin{equation}
\BBE:=\bigcup_{\alpha=x,z}\bigcup_{\begin{array}{c}\scriptstyle\,\langle\br,\br'\rangle_\alpha\\*[-6pt]\scriptstyle\br,\br'\in\Lambda_B\end{array}} 
\Bigl\{\bS\in\Bcal\colon\,|S_{\br}^{\alpha}-S_{\br'}^{\alpha}|\ge \frac{\Delta}{16B}\Bigr\}.
\end{equation}
These will be the configurations with too much energy  ---  hence the subscript ``E''. 
The remaining ``bad'' configurations will be collected in the set
\begin{equation}
\BBSW:=\Bcal\setminus\BBE,
\end{equation}
where ``SW'' designates the method  ---  a \emph{spin-wave} calculation  ---  that will be used to estimate the contributions to this event.
The following claim provides a key structural information on the configurations contained in $\BBSW$:

\begin{lemma}
\label{L:BBSW-structura}
Suppose that $\eta$ and $\Delta$ satisfy the inequalities
\begin{equation}
\label{E:eta-Delta-bds}
\Delta<\eta^2<1.
\end{equation}
Then for each $\bS\in \BBSW$, there is a unit vector $\bv\in\R^d$ and a set $\Lambda\subset\Lambda_B^0$ such that
\begin{equation}
\bigl|\bigl(\bvarphi_\Lambda(\bS)\bigr)_{\br}-\bv\bigr|\le \frac\Delta2,\qquad \br\in\Lambda_B .
\end{equation}
\end{lemma}

\begin{proof}{Proof}
Fix $\bS\in\BBSW$ and define $\Lambda\subset \Lambda_B^0$  by taking $\Lambda\cap\Lambda_B^{\a}=\{\br\in\Lambda_B\colon S_{\br}^{\a}<0\}$, $\a=1,2$. Consider the configuration $\tilde{\bS}:=\bvarphi_\Lambda(\bS)$. Clearly, $\tilde{S}_{\br}^{\a}\ge 0$ for all $\br\in \Lambda_B^{\a}$, $\alpha=x,z$. 
The proof of the main claim now comes in two steps. First we will show that both components of all $\bS_{\br}$, $\br\in\Lambda_B$, are at least $\sqrt\zeta$ in absolute value, where 
\begin{equation}
\zeta:=\frac{\eta^2}2-\frac\Delta4.
\end{equation}
This will be used to ensure that $\tilde{S}_{\br}^{\a}\ge 0$ everywhere. Then we will rerun part of the argument to show that $\tilde{\bS}$ is, actually, to within $O(\Delta)$ of a constant configuration.

Let $\zeta$ be as above and abbreviate $c:=(16B)^{-1}$. Notice that the bounds \eqref{E:eta-Delta-bds} imply
\begin{equation}
\zeta+\frac\Delta4\le\frac12,\,\,\,c\Delta<\frac14\quad\mbox{and}\quad\sqrt\zeta>2c\Delta.
\end{equation}
 First, we claim that
\begin{equation}
\label{E:LB-on-component}
\bS\in\BBSW\quad\Rightarrow\quad[S_{\br}^{\alpha}]^2\ge\zeta,\qquad\br\in\Lambda_B,\,\,\alpha=x,z.
\end{equation}
Focusing our attention on  on $\alpha=1$, assume that $[S_{\br}^{x}]^2<\zeta$ at some $\br\in\Lambda_B$ and derive a contradiction with $\bS\in\BBSW$.
Notice that $\bS\in\BBSW$ implies that 
\begin{equation}
\bigl|S_{\br}^{\alpha}-S_{\br'}^{\alpha}\bigr|<c\Delta
\end{equation}
for any $\br,\br'\in\Lambda_B$ connected by an $\alpha$-edge.
If $\br'$ is a neighbor of $\br$ over a $x$-edge, then
\begin{equation}
\label{E:3.63}
[S_{\br'}^{x}]^2<[\sqrt\zeta+c\Delta]^2\le\zeta+2c\Delta,
\end{equation}
where we used that $\zeta\le\frac12$ and $c\Delta\le\frac12$.
On the other hand, if $\br'$ is connected to $\br$ by a $z$-edge, we can proceed via the second components to get again
\begin{eqnarray}
[S_{\br'}^{x}]^2&=1-[S_{\br'}^{z}]^2
\\
&\le1-[S_{\br}^{z}]^2+2|S_{\br'}^{z}-S_{\br}^{z}|
\\
&<\zeta+2c\Delta.
\end{eqnarray}
Examining all pairs of nearest neighbors along a shortest path from $\br$ to any $\br'\in\Lambda_B$ and using that this path has at most $2B$ edges, we conclude
\begin{equation}
[S_{\br'}^{x}]^2<\zeta+(2B)2c\Delta=\zeta+\frac\Delta4,\qquad\br'\in\Lambda_B.
\end{equation}
Notice that to maintain the argument from  \eqref{E:3.63} for the proof of the inequality $[S_{\br''}^{x}]^2<[S_{\br'}^{x}]^2+2c\Delta$ for every step $(\br',\br'')$ of the iteration,
we actually need to invoke that  $\zeta+\frac\Delta4\le\frac12$.
At the same time, since $\sqrt{1-(\zeta+\frac\Delta4)}\ge\frac12\ge2c\Delta$, the triangle inequality implies $\tilde S_{\br'}^{z}\ge0$ on each $z$-plaquette, and thus everywhere in~$\Lambda_B$. In addition, this allows us to compute further:
\begin{eqnarray}
\bigl|\tilde{\bS}_{\br'}-\hate_2\bigr|^2&=[S_{\br'}^{x}]^2+\Bigl(1-\sqrt{1-[S_{\br'}^{x}]^2}\Bigr)^2
\\
&\le [S_{\br'}^{x}]^2+[S_{\br'}^{x}]^4\le 2[S_{\br'}^{x}]^2<2\Bigl(\zeta+\frac\Delta4\Bigr)=2\zeta+\frac\Delta2=\eta^2,
\end{eqnarray}
where we used that $1-\sqrt{1-x}\le x$ for $0<x<1$. This would imply that $\tilde{\bS}\in\Gcal_z^0$ and thus $\bS\in\Gcal_z$, in contradiction with the assumption that $\bS\in\BBSW\subset\Bcal$.

Having proven the bound \eqref{E:LB-on-component}, we can use the fact that $\sqrt\zeta>2c\Delta$ to imply that $\tilde S_{\br}^{\alpha}\ge0$ for all $\br\in\Lambda_B$. 
Combining now the positivity of components of  $\tilde S_{\br}^{\alpha}$ with \eqref{E:LB-on-component}, we can improve the bound on the difference of the 1-components of spins over a $z$-edge $\langle\br,\br'\rangle_z$ as follows: Suppose without loss of generality that $S_{\br'}^{z}\ge S_{\br}^{z}$ and write
\begin{eqnarray}
\bigl|S_{\br'}^{x}-S_{\br}^{x}\bigr|&=\Bigl|\sqrt{1-[S_{\br'}^{z}]^2}-\sqrt{1-[S_{\br}^{z}]^2}\Bigr|\\
&\le \frac{|S_{\br'}^{z}-S_{\br}^{z}|}{\sqrt{1-[S_{\br'}^{z}]^2}}\le\frac{c\Delta}{\sqrt{1-\zeta}}
\end{eqnarray}
In conjunction with $|S_{\br'}^{z}-S_{\br}^{z}|<c\Delta$, this implies
\begin{equation}
\bigl|\bS_{\br'}-\bS_{\br}\bigr|\le c\Delta\frac{2-\zeta}{1-\zeta}\le 4c\Delta
\end{equation}
where we used that $1-\zeta\ge\frac12$.
As $(2B)4c\Delta<\frac\Delta2$, the desired claim now follows with $\bv:=\tilde{\bS}_0$ and $\Lambda$ as above by examining a path 
of minimal length between the origin and any other site in $\Lambda_B$.
\end{proof}

The previous lemma allows us to further partition $\BBSW$ as follows. Let $\bv_1,\dots,\bv_n$ denote the unit vectors representing the complex $n$-th roots of unity, 
$\bv_\ell:= \be(\ell\frac{2\pi}{n})$, $\ell=1,\dots,n$,
where $n:=\lfloor\frac{2\pi}{\Delta}\rfloor+1$ is the smallest integer for which $n\Delta>2\pi$. Defining
\begin{equation}
\BBSW^{(\ell)}:=\bigcup_{\Lambda\subset\Lambda_B^0}\bvarphi_\Lambda\Bigl(\bigl\{\bS\in\BBSW\colon\forall \br\in\Lambda_B\,\,|\bS_{\br}-\bv_{\ell}|<\Delta\bigr\}\Bigr),
\end{equation}
it follows that
\begin{equation}
\BBSW=\bigcup_{\ell=1}^n\BBSW^{(\ell)}.
\end{equation}
Having in mind that $\Delta\ll \eta$ and that, by definition,
$\BBSW\cap (\Gcal_x\cup\Gcal_z)=\emptyset$, 
we notice that  $\BBSW^{(\ell)}=\emptyset$
whenever the distance of $\bv_\ell$ from the points $\pm\hate_1,\pm\hate_2$, is less than $\eta-\Delta$.
The set function $\Acal\mapsto\mathfrak p_\beta(\Acal)$ is subadditive (see~\cite[Lemma~5.9]{Biskup-book}) and so we have
\begin{equation}
\label{E:subadditive}
\mathfrak p_\beta(\Bcal)\le\mathfrak p_\beta(\BBE)+\sum_{\ell=1}^n\mathfrak p_\beta(\BBSW^{(\ell)}).
\end{equation}
It remains to derive suitable estimates on $\mathfrak p_\beta(\BBE)$ and $\mathfrak p_\beta(\BBSW^{(\ell)})$. 
An indispensable ingredient will be the following lower bound on the full partition function:

\begin{lemma}
\label{L:LB-PF}
Fix $\tau>0$ such that $\sin(2\tau)>\tau$ and let $\delta$ and $N_0$ be as in Proposition~\ref{P:Gaussian}. If $\Delta$ and $\beta$ satisfy
the conditions \eqref{Delta-cond} and $N\ge N_0$, then
\begin{equation}
\bigl( Z_{N,\b}\bigl)^{1/|\T_N|}\ge  \texte^{\,\b J - F(\tau) -\tau},
\end{equation}
where $F$ is the free energy introduced in \eqref{E:3.19}.
\end{lemma}

\begin{proof}{Proof}
By restricting the integrals over the spins to the set where $|\bS_{\br}-\hate(\tau)|<\Delta$, we get
\begin{equation}
Z_{N,\beta}\ge \texte^{\beta J|\T_N|}\mathcal Z_N(\tau,\Delta).
\end{equation}
From here the claim follows by invoking the bound \eqref{E:freeenapprox}.
\end{proof}

\begin{lemma}
Suppose that $\eta\in(0,\ffrac\pi4)$ and, given ${\tau}>0$ with $\sin(2\tau)>\tau$, let $\delta$ be as in Proposition~\ref{P:Gaussian} and assume that $\Delta$ and $\beta$ satisfy the conditions \eqref{Delta-cond} and that $\Delta< \eta^2<(\tfrac{\pi}4)^2$.
Then
\begin{equation}
\label{E:BEbound}
\mathfrak p_\beta(\BBE) \le\,(2B)^2\,\texte^{\,c_1B^2-c_2\beta J\Delta^2/B^2},
\end{equation}
where  $c_1:=\frac12\log(2\pi)+F(\tau)+\tau$ and $c_2:=1/512$, and
\begin{equation}
\label{E:BSWbound}
\mathfrak p_\beta(\BBSW^{(\ell)}) \le   2^{\ffrac14}\,\texte^{-[F(\eta-\Delta)-F(\tau)-\tau]B^2},  \quad \ell=1,\dots,n.
\end{equation}
\end{lemma}

\begin{proof}{Proof}
Consider the disseminated event $\overline{\BBE}:=\bigcap_{\br\in\T_{N/B}}\theta_{\br}(\BBE)$. We need to derive an upper bound on the constrained partition function $Z_{N,\beta}(\overline{\BBE})$, which is given by the same integral as the full partition function but only over configurations in the event $\overline{\BBE}$. Since on $\BBE$, each configuration has an ``energetically charged'' edge in each block in $\{\Lambda_B+B\bt\colon\bt\in\T_{N/B}\}$, Lemma~\ref{lemma-H-rewrite} tells us
\begin{equation}
\mathcal H_N(\bS)\ge\frac12\Bigl(\frac{\Delta}{16B}\Bigr)^2 J (N/B)^2-J|\T_N|,\qquad \bS\in\overline{\BBE}.
\end{equation}
To account for the entropy, we note that there are altogether $2B(B-1)$ positions where the ``energetically charged'' edge can occur in each translate of $\Lambda_B$. Lemma~\ref{L:LB-PF} now gives
\begin{eqnarray}
\mu_{N,\beta}(\overline{\BBE})^{(B/N)^2}&=\left(\frac{Z_{N,\beta}(\overline{\BBE})}{Z_{N,\beta}}\right)^{(B/N)^2}
\\
&\le (2\pi)^{B^2/2}(2B)^2\texte^{-\beta J\Delta^2/512+[F(\tau)+\tau]B^2},
\end{eqnarray}
where the factor $(2\pi)^{B^2/2}$ comes from performing (unconstrained) integrals of the spins.
The limit $N\to\infty$ along multiples of~$B$ now yields \eqref{E:BEbound}.

Next we consider the event $\overline{\BBSW^{(\ell)}}:=\bigcap_{\br\in\T_{N/B}}\theta_{\br}(\BBSW^{(\ell)})$. We only need to focus on $\ell$ satisfying $|\frac{2\pi}n\ell-\frac{2\pi}nk|\ge\eta-\Delta$ for $k=0,1,2,3$ because otherwise the event is void. Here we notice that, for each $\bS\in\overline{\BBSW^{(\ell)}}$, there is a unique $\Lambda\subset \bigcup_{\bt\in\T_{N/B}}(B\bt+\Lambda_B^0)$ such that the configuration $\tilde{\bS}:=\bvarphi_\Lambda(\bS)$ is within $\Delta$ of the vector $\bv_\ell$. The number of distinct $\Lambda$ associated with a single $\tilde{\bS}$ is $2^{(N/B)^2/4}$. As both the \emph{a priori} measure and the Hamiltonian are invariant under $\bvarphi_\Lambda$, this yields
\begin{equation}
Z_{N,\beta}\bigl(\overline{\BBSW^{(\ell)}}\bigr)\le2^{(N/B)^2/4}\mathcal Z_N\bigl(\ell{\textstyle\frac{2\pi}n},\Delta\bigr).
\end{equation}
Hence,
\begin{eqnarray}
\mu_{N,\beta}\bigl(\overline{\BBSW^{(\ell)}}\bigr)^{(B/N)^2}&=\left(\frac{Z_{N,\beta}\bigl(\overline{\BBSW^{(\ell)}}\bigr)}{Z_{N,\beta}}\right)^{(B/N)^2}
\\
&\le 2^{\ffrac14}\exp\Bigl\{\bigl[-F\bigl(\ell{\textstyle\frac{2\pi}n}\bigr)+F(\tau)+\tau\bigr]B^2\Bigr\}.
\end{eqnarray}
Now we invoke Proposition~\ref{P:propertiesofM} and apply the monotonicity and periodicity properties of~$F$ to infer that
$ F(\ell \frac{2\pi}n)\ge F(\eta-\Delta)$ for any $\ell$ for which $\BBSW^{(\ell)}\ne\emptyset$. (This is where we need $0<\Delta<\eta<\ffrac\pi4$.) The limit $N\to\infty$ (along multiples of~$B$) finishes the claim.
\end{proof}

\begin{proof}{Proof of Proposition~\ref{P:smallBcal}}
Fix~$\epsilon>0$ and $\eta\in(0,\ffrac\pi4)$ and set $\tau$ to be a positive number such that $\sin(2\tau)>\tau$ and
\begin{equation}
\label{E:tau-choice}
F\bigl(\ffrac\eta2\bigr)>F(\tau)+\tau.
\end{equation}
This is possible because $\theta\mapsto F(\theta)$ tends to minus infinity as $\theta\downarrow0$, see Fig.~2.
Let $\delta$ be related to this $\tau$ as in Proposition~\ref{P:Gaussian}. We will henceforth link $\Delta$ and $B$ to $\beta>0$ via
\begin{equation}
\label{E:Delta-B}
\Delta:=\beta^{-\frac5{12}}\quad\mbox{and}\quad B:=2\lfloor\log\beta\rfloor.
\end{equation}
Notice that this choice of $\Delta$ will make the bounds in Proposition~\ref{P:Gaussian} true once $\beta$ is sufficiently large. Let $\beta_0$ be now a value such that for all $\beta\ge\beta_0$ these bounds hold, the condition \eqref{E:eta-Delta-bds} in Lemma~\ref{L:BBSW-structura} is satisfied, and the inequalities
\begin{equation}
\label{E:eta-D}
\eta-\Delta>\frac\eta2
\end{equation}
and
\begin{equation}
\label{E:econd}
(2B)^2\,\texte^{\,c_1B^2-c_2\beta J\Delta^2/B^2}+\frac{4\pi}{\Delta}\,
2^{\ffrac14}\,\texte^{-[F(\ffrac\eta2)-F(\tau)-\tau]B^2}<\e
 \end{equation}
hold true. This is possible by  \eqref{E:tau-choice} and our choices \eqref{E:Delta-B}  ---  and the fact that the quantities~$c_1$ and $c_2$ depend only on $\tau$.

We claim that \eqref{E:econd} implies the desired bound $\mathfrak p_\beta(\Bcal)<\e$. Indeed, using that $n\le\ffrac{4\pi}\Delta$ for sufficiently small  $\Delta$, it suffices to insert into \eqref{E:subadditive} the bound on $\mathfrak p_\beta(\BBE)$ and the uniform bound on $\mathfrak p_\beta(\BBSW^{(\ell)})$ from \twoeqref{E:BEbound}{E:BSWbound}.  As a result we get the estimate of $\mathfrak p_\beta(\Bcal)$ by the left-hand side of \eqref{E:econd}. (We also used that $F$ is strictly increasing on $(0,\ffrac\pi4)$ and applied \eqref{E:eta-D}.) By \eqref{E:econd}, $\mathfrak p_\beta(\Bcal)$ is thus less than $\e$ for any~$\beta$ in excess of $\beta_0$ defined above.
\end{proof}

\section{Proofs of main results}
\label{sec-V}
\subsection{Classical model}
\label{sec-V.1}
We are now ready to prove our main results for the classical system. The key inputs are chessboard estimates alongside with the bound in Proposition~\ref{P:smallBcal}.

\begin{proof}{Proof of Theorem~\ref{T:LRO}}
The standard line of reasoning leading to the proof of phase coexistence in the present context is based on the observation
that distinct types of good events are unlikely to occur in the same configuration. A formal statement is as follows:

\begin{quote}
Assume that $B\ge 4$ and $\eta <  2\sin(\pi/8)$. For any $\epsilon>0$ there exist $\delta>0$ and $N_0$ such that if  $\mathfrak p_\beta(\Bcal)<\delta$ and $N\ge N_0$,
then $\mu_{N,\beta}(\theta_{\br}(\Gcal_x)\cap\theta_{\br'}(\Gcal_z))\le \epsilon$
for all $\br,\br'\in  \T_{N/B}$.
\end{quote}

\noindent
The proof of this fact is an application of a Peierls-type contour argument. 
Fix $\br,\br'\in  \T_{N/B}$, $\br\neq\br'$, and let $\mathcal{Y}(\br)$ denote the set of all 
$\Lambda\subset \T_{N/B}$ such that both $\Lambda$ and $\T_{N/B}\setminus \Lambda$ are connected
and $\Lambda\ni \br$.  We use $\partial\Lambda$ to denote the set of all vertices outside, but adjacent to $\Lambda$.
For $\Lambda\in\mathcal{Y}(\br)$, consider the event
\begin{equation}
R_\Lambda:= \Bigl\{  \sum_{\br'\in\partial\Lambda}1_{\Bcal}\circ \theta_r > \tfrac15  \abs{\partial\Lambda} \Bigr\},
\end{equation}
where $\theta_{\br}(\Bcal)=\vartheta_{\br}(\Bcal)$ with $\theta_{\br}$ denoting the shift by $B\br$. Recalling also the notation $\sigma_1$, resp., $\sigma_2$
for the shift by $2\be_1$, resp., $2\be_2$, we now claim
\begin{equation}
\theta_{\br}(\Gcal_x)\cap\theta_{\br'}(\Gcal_z)\subset \bigcup_{\begin{array}{c}\scriptstyle\,\Lambda\in\mathcal{Y}(\br)\\*[-6pt]\scriptstyle\br'\not\in \Lambda\end{array}} \bigl(R_{\Lambda} \cup
\bigcup_{j=1,2} \bigl(\sigma_j(R_\Lambda)\cup \sigma^{-1}_j(R_\Lambda)\bigr)_{\br}\bigr).
\end{equation}
To prove this, consider a configuration $\bS\in\theta_{\br}(\Gcal_x)\cap\theta_{\br'}(\Gcal_z)$ and let $\Lambda'$ denote the set of all $\bs\in 
\T_{N/B}$ for which there is a nearest-neighbor path that starts at $\br$, ends at $\bs$, and visits only blocks $B\bs'+\Lambda_B$, where $\theta_{\bs'}(\Gcal_x)$ occurs.
Clearly, $\br'\not\in\Lambda'$ so we may use $\L$ to denote the set of all $\bs\in 
\T_{N/B}$ such that every nearest neighbor path from $\bs$ to $\br'$ contains at least one vertex of $\L'$ --- i.e.
$\L$ is $\L'$ with all of its ''holes'' filled.

Our task is to show that 
\begin{equation}
\label{E:RsR}
\bS\in R_{\Lambda} \cup
\bigcup_{j=1,2} \bigl(\sigma_j(R_\Lambda)\cup \sigma^{-1}_j(R_\Lambda)\bigr)_{\br}.
\end{equation}
If $\bS\in R_{\L}$, then we are done, so let us suppose that $\bS\not\in R_{\L}$.
Under this condition more than $4/5$ of all blocks corresponding to vertices in $\p\L$ are good but, since they are not part of $\L$, they are of type $\Gcal_z$. Then, however, more than $2/5$ of $\Gcal_z$-blocks in $\p\L$ are adjacent to $\L$ in one coordinate direction --- say it is the first one --- and thus $1/5$ of them are adjacent in either positive or negative coordinate direction.
Observing that $\bS\in  \theta_{\br}(\Gcal_x)\cap\theta_{\br+\be_j}(\Gcal_z)$ implies that
$\sigma^{-1}_1(\bS)\in \theta_{\br}(\Bcal)$ --- and similarly for the opposite direction --- we get~\eqref{E:RsR}.

With \eqref{E:RsR} on hand, we now perform a standard version of the Peierls argument combined with chessboard estimates. A key input is the fact that, for some $c\in(1,\infty)$,
\begin{equation}
\abs{\{\L\in\Ycal(\br): \abs{\p\L}=n\}}\le c^n.
\end{equation}
Using inclusion-exclusion and $\sigma_j$-invariance of $\mu_{N,\beta}$, it now suffices to show that for some $\tilde c<\infty$ and $N_0<\infty$,
\begin{equation}
\label{E:muR<}
\mu_{N,\beta}( R_{\Lambda})\le \tilde c^{\abs{\p\L}} \mathfrak p_\beta(\Bcal)^{\abs{\p\L}/5},
\qquad\Lambda\in\Ycal(\br).
\end{equation}
To this end, let $N_0$ be such that, for all $N\ge N_0$, the probability on the right hand side of \eqref{E:pbeta}
for $\Acal:=\Bcal$ is at most $\bigl[2  \mathfrak p_\beta(\Bcal)\bigr]^{\T_N}$. Now, cover $R_{\L}$ by the union over all subsets of $\p\L$ with $\abs{\p\L}/5$ elements where the bad event occurs.
The number of such partitions is at most $2^{\abs{\p\L}}$;
the probability of each occurrence is estimated by $\bigl[2  \mathfrak p_\beta(\Bcal)\bigr]^{\abs{\p\L}/5}$.
This proves \eqref{E:muR<} with $\tilde c:=2^{6/5}$ and thus the above claim.

Applying this claim alongside the fact that, on a good event, the spins are oriented along one of the coordinate directions proves \eqref{E:musym} and \eqref{E:ES^2x}. To get also \eqref{E:ES^2}, one just follows literally the argument proving the main result in~\cite{BK2}.
\end{proof}

\begin{proof}{Proof of Theorem~\ref{T:away-from-symmetry}}
This is a consequence of Theorem~\ref{T:LRO} and a convexity argument. Let $\Ecal^{x}=\Ecal_0$ denote the Wenzel-Janke plaquette energy \eqref{E:WJen} for the type-$x$ plaquette at the origin. Let $\mathfrak G(J_1,J_2)$  denote the set of all translation-invariant, ergodic Gibbs states. The convexity of $J_1\mapsto\log Z_N$ then implies that once $J_1<J_1'$, for any $\mu\in \mathfrak G(J_1,J_2)$ and $\mu'\in \mathfrak G(J_1',J_2)$ we have
\begin{equation}
\E_\mu(\Ecal^{x})\le \E_{\mu'}(\Ecal^{x}).
\end{equation}
Now, at $J_1=J_2=:J$ and $\b\ge\b_0$, Theorem~\ref{T:LRO} guarantees the existence of a
$\mu_\beta^{x}\in \mathfrak G(J,J)$, such that 
\begin{equation}
\E_{\mu_\beta^{x}}(\Ecal^{x})\ge 4-\delta.
\end{equation}
(Here~$\delta$ may differ from the one used in the statement of Theorem~\ref{T:LRO}.) Thus, for all $J_1>J_2$ and $\mu\in \mathfrak G(J_1,J_2)$, we have 
\begin{equation}
\E_{\mu}(\Ecal^{x})\ge 4-\delta.
\end{equation}
As $\Ecal^{x}$ is the sum of four terms of the type $S_{\br}^{x}S_{\br'}^{x}$, which are less than one, we must have
\begin{equation}
\E_{\mu}(S_{\br}^{x}S_{\br'}^{x})\ge 1-\delta
\end{equation}
for all nearest-neighbour pairs $\langle \br, \br' \rangle_x$ of type $x$. This forces $\E_\mu(|S_{\br}^{x}|)\ge1-\delta$ and thus $\E_\mu([S_{\br}^{x}]^2)\ge1-2\delta$. As~$\delta$ is arbitrary, this proves the claim.
\end{proof}

We will also give the formal proof of absence of magnetic order:

\begin{proof}{Proof of Theorem~\ref{T:no-magnet}}
It is easy to check that every Gibbs measure is invariant under the action of any $\bvarphi_{\br}$. As each spin belongs to one plaquette of type-$x$ and one plaquette of type-$z$, it follows that the distribution of $S_{\br}^{\alpha}$ is symmetric around zero, for any~$\br$ and any~$\alpha=x,z$. This implies the claim.
\end{proof}

\subsection{Quantum model}
\label{sec-V.2}
Our set of results for the quantum model will be derived by an application of the general theory developed in \cite{Biskup-Chayes-Starr} whose main conclusion can be found in Theorem~3.7 of~\cite{Biskup-Chayes-Starr}. This theorem says roughly the following: Whenever a quantum-spin model satisfies the conditions of (quantum) reflection positivity, and the classical system admits a proof of phase coexistence at a positive temperature by means of chessboard estimates, then the same phase coexistence occurs in the quantum system provided the magnitude of the quantum spin is sufficiently larger than the inverse temperature squared.

First let us check that the prerequisite concerning the quantum reflection positivity is satisfied. Since we are using reflections in planes bisecting edges of~$\T_N$, this is proved by the same argument as in the classical case, except that we need to write all operators in a basis in which their matrix elements are all real valued. This is satisfied automatically in the representation \eqref{E:Spm} in which~$\hat S_{\br}^z,\hat S_{\br}^\pm$ are real and so is~$\hat S_{\br}^x=\frac12(\hat S_{\br}^++\hat S_{\br}^-)$.

All we have to do is thus adapt the proof for the classical model to plug into Theorem~3.7 of~\cite{Biskup-Chayes-Starr}. We begin by introducing the formalism of \emph{coherent states} that the whole connection is based on. Consider the space~$\mathbb C^{2\Scal+1}$ that carries the corresponding $2\Scal+1$-dimensional representation of $\su(2)$. Let~$\Omega$ be a vector on the unit sphere $\mathscr S_2$ in~$\mathbb R^3$ that is described by the spherical angles $\theta$ and $\phi$. Then we set
\newcommand\ket[1]{|#1\rangle}
\begin{equation}
\ket{\Omega}:= \sum_{M=-\Scal}^{\Scal}{{2\Scal}\choose{\Scal+M}}^{\ffrac12}\,\,
\bigl[\cos(\ffrac\theta2)\bigr]^{\Scal+M}\,
\bigl[\sin(\ffrac\theta2)\bigr]^{\Scal-M}\,
   \texte^{\texti(\Scal-M)\phi}\, \ket{M}.
\end{equation}
Abusing the notation slightly, whenever $\Omega:=(\Omega_{\br})$ is a collection of such vectors, we will denote the corresponding product state by $\ket{\Omega}:=\bigotimes_{\br}\ket{\Omega_{\br}}$.

The coherent states have a number of remarkable properties of which relevant for us are particularly those listed in Sect.~2.1 of~\cite{Biskup-Chayes-Starr}. Here we will only need the notions of the \emph{lower} and \emph{upper symbols}. Given a linear operator $\hat A$ on $\bigotimes_{\br\in\Lambda}\mathbb C^{2\Scal+1}$, the lower symbol is the function $\Omega\mapsto\langle \hat A\rangle_\Omega$ on $(\mathscr S_2)^\Lambda$ such that
\begin{equation}
\langle \hat A\rangle_\Omega:=\langle\Omega|\hat A|\Omega\rangle.
\end{equation}
The upper symbol is, in turn, a function $\Omega\mapsto[\hat A]_\Omega$ such that
\begin{equation}
\hat A=\left(\frac{2\Scal+1}{4\pi}\right)^{|\Lambda|}\int_{(\mathscr S_2)^{\Lambda}}\textd\Omega\,\,[\hat A]_\Omega\,|\Omega\rangle\langle\Omega|,
\end{equation}
where $\textd\Omega:=\prod_{\br\in\Lambda}\textd\Omega_{\br}$ with $\textd\Omega_{\br}$ denoting the uniform measure on $\mathscr S_2$ with total mass~$4\pi$. The upper symbol is not necessarily unique, so we use the notation $[\hat A]_\Omega$ to denote any version thereof. The upper and lower symbols are two natural classical approximations of the quantum Hamiltonian, so we need to check that they are reasonably close:

\begin{lemma}
\label{lemma-q-class}
Consider the operator $\hat{\mathcal H}_N$ in \eqref{E:quantum-Ham}. Then there is a constant $c=c(J_1,J_2)<\infty$ and a version of the upper symbol $[\hat{\mathcal H}_N]_\Omega$ such that for each $\Omega\in(\mathscr S_2)^{\T_N}$,
\begin{equation}
\label{E:upper-lower}
\bigl|\langle\hat{\mathcal H}_N\rangle_\Omega-[\hat{\mathcal H}_N]_\Omega\bigr|\le \frac c{\Scal}|\T_N|.
\end{equation}
\end{lemma}

\begin{proof}{Proof}
Let $\mathcal H_N(\bS)$ denote the classical Hamiltonian corresponding to \eqref{E:quantum-Ham}  ---  obtained by replacing  the operators 
$\hat{\bS}_{\br}$  by vectors $\bS_{\br}\in\mathscr S_2$ and dropping the normalization by $\Scal^2$.
Then the desired bound will hold once we verify that
\begin{equation}
\langle\hat{\mathcal H}_N\rangle_\Omega= \mathcal H_N(\Omega)\quad\mbox{and}\quad
[\hat{\mathcal H}_N]_\Omega=(1+1/\Scal)^2 \mathcal H_N(\Omega)
\end{equation}
for some version of $[\hat{\mathcal H}_N]_\Omega$. This is in turn shown by noting that $\hat{\mathcal H}_N$ is multilinear in the operators $
\hat{\bS}$  ---  which means that for each term in the sum the corresponding symbols take the form of a product  ---  and by the fact that $\langle
\hat{\bS}
_{\br}\rangle_\Omega=\Scal\Omega_{\br}$ and that $[
\hat{\bS}
_{\br}]_\Omega:=(\Scal+1)\Omega_{\br}$ is a version of the upper symbol for $
\hat{\bS}
$. The bound \eqref{E:upper-lower} now follows by the fact that $\mathcal H_N(\Omega)$ is bounded by a constant times $|\T_N|$, uniformly in~$\Omega$ and~$\Scal\ge\ffrac12$.
\end{proof}

Our next step will be to  verify the conditions required by \cite[Theorem~3.7]{Biskup-Chayes-Starr} for the corresponding classical model. The problem here is that the link through the coherent states naturally leads to \emph{three}-component classical spins. Consider the events $\widetilde\Gcal_x, \widetilde\Gcal_z$ and $\Bcal$ that are defined as follows. For a three-component spin configuration $\bS=(\bS_{\br})$ with $\bS_{\br}=(S_{\br}^x,S_{\br}^y,S_{\br}^z)$, let $\bS^{xz}$ denote its projection onto the $xz$-plane \emph{scaled} to have a unit length. For a $B$-block event~$\Acal\subset(\mathscr S_2)^{\T_N}$, let 
\begin{equation}
\widetilde\Acal:=\bigl\{\bS\colon\bS^{xz}\in\Acal\bigr\}\cap\bigcap_{\br\in\Lambda_B}\bigl\{\bS\colon |S_{\br}^{y}|
\le\Delta\bigr\}
\end{equation}
denote its natural extension to three-component spin configurations. This immediately defines the  events $\widetilde\Gcal_x$, $\widetilde\Gcal_z$, and $\widetilde\Bcal$. However, $\widetilde\Bcal$ does not cover the complement $\widetilde\Gcal_x\cup\widetilde\Gcal_z$; for that we will also need the event
\begin{equation}
\Bcal_y:=\bigcup_{\br\in\Lambda_B}\bigl\{\bS\colon |S_{\br}^{y}|>\Delta\bigr\}.
\end{equation}
Then, obviously, $\Bcal:=(\widetilde\Gcal_x\cup\widetilde\Gcal_z)^{\rm c}$ satisfies $\Bcal = \widetilde\Bcal\cup\Bcal_y$. In order to plug into our calculations for the two-component spin, we will need to prove the two lemmas.

\begin{lemma}
\label{lemma-5.2}
Suppose~$\Delta<\ffrac12$. Then there exists a constant~$c_3<\infty$ such that, for any~$B$-block event~$\Acal$ such that $\Acal\cap\BBE=\emptyset$, its counterpart~$\widetilde\Acal$ defined as above satisfies
\begin{equation}
\mathfrak p_\beta(\widetilde\Acal)\le \Bigl(\pi\Delta^{-2}\,\texte^{\beta J c_3 \Delta^3}\Bigr)^{B^2}
\frac{\mathfrak p_\beta(\Acal)}{1-\mathfrak p_\beta(\BBE)}
\end{equation}
\end{lemma}

\begin{proof}{Proof}
Consider the expectation of the disseminated event~$\widetilde\Acal$ in the torus measure~$\tilde\mu_{N,\beta}$ for the three-component model. We will write this expectation as~${\widetilde Z}_{N,\beta}(\widetilde\Acal)/{\widetilde Z}_{N,\beta}$. Pick $\tilde{\bS}\in\widetilde\Acal$ and define $\bS:=\tilde{\bS}^{xz}$. First, by the very definition $\bS$ belongs to the disseminated event~$\Acal$. So in order to compare the expectation for the three-component spins with that for the two-component spins, we will need to control the change in the Hamiltonian and the \emph{a priori} measure under the map $\tilde{\bS}\mapsto\bS$. 

To treat both ${\widetilde Z}_{N,\beta}(\widetilde\Acal)$ and~${\widetilde Z}_{N,\beta}$ in a unified fashion, assume that $|\tilde S_{\br}^{y}|\le \delta$ for some~$\delta\le\Delta$ and all $\br\in \T_N$. Let~$\bS:=\tilde{\bS}^{xz}$. Note that
\begin{equation}
\tilde S_{\br}^{\alpha}=\frac{S_{\br}^{\alpha}}{\sqrt{1-[\tilde S_{\br}^{y}]^2}}=S_{\br}^{\alpha}+O(\delta^2),\qquad\alpha=x,z.
\end{equation}
From the fact that $\Acal\cap\BBE=\emptyset$ it follows that
\begin{equation}
(\tilde S_{\br}^{\alpha}-\tilde S_{\br'}^{\alpha})^2=(S_{\br}^{\alpha}-S_{\br'}^{\alpha})^2+O(\delta^2\Delta)
\end{equation}
and so we have
\begin{equation}
\label{E:5.20}
\mathcal H_N(\tilde{\bS})=\mathcal H_N(\bS)+J\sum_{\br\in\T_N}[\tilde S_{\br}^{y}]^2+O(\delta^2\Delta)J|\T_N|.
\end{equation}
The \emph{a priori} measures are related by $\tilde\nu(\textd\tilde S_{\br})=\sqrt{1-[\tilde S_{\br}^{y}]^2}\,\,\nu(\textd S_{\br})\textd\tilde S_{\br}^y$.

We will now derive  bounds on the  partition functions ${\widetilde Z}_{N,\beta}(\widetilde\Acal)$ and ${\widetilde Z}_{N,\beta}$
in terms of their two-component spin counterparts $ Z_{N,\beta}(\Acal)$ and $Z_{N,\beta}$.
For ${\widetilde Z}_{N,\beta}(\widetilde\Acal)$ we set $\delta:=\Delta$ and note that 
$\mathcal H_N(\tilde{\bS})\ge\mathcal H_N(\bS)+ O(\Delta^3) \beta J \abs{\T_N}$. Integrating out the components 
$\tilde S_{\br}^{y}$ gives
\begin{equation}
{\widetilde Z}_{N,\beta}(\widetilde\Acal)\le  Z_{N,\beta}(\Acal)  \bigl(\pi\texte^{O(\Delta^3)\beta J}\bigr)^{\abs{\T_N}}.
\end{equation}
For a lower bound on~${\widetilde Z}_{N,\beta}$, we first note that ${\widetilde Z}_{N,\beta}\ge {\widetilde Z}_{N,\beta}(\BBE^{\rm{c}}\cap\Gcal_y')$, where $\Gcal_z'$ is the event 
\begin{equation}
\Gcal_y'=  \bigcap_{\br\in\Lambda_B}\bigl\{\bS : \abs{S_{\br}^{y}}<\delta\bigr\}
\end{equation}
with $\delta:=\Delta^2$.  Applying \eqref{E:5.20}, the Hamiltonians now differ by $O(\Delta^3) J\abs{\T_N}$ and the integral over $\tilde S_{\br}^{y}$, $\br\in\T_N$, now yields a term 
$\Delta^2$ per site. This shows
\begin{equation}
{\widetilde Z}_{N,\beta}\ge \bigl(\Delta^{2}\texte^{O(\Delta^3)\beta J}\bigr)^{\abs{\T_N}} Z_{N,\beta}(\BBE^{\rm{c}}).
\end{equation}
Combining the upper and lower bounds and applying the subadditivity bound $\mathfrak p_\beta(\BBE^{\rm c})\ge1-\mathfrak p_\beta(\BBE)$ we get the desired claim.
\end{proof}

\begin{lemma}
\label{lemma-5.3}
There are constants~$c_4,c_5\in(0,\infty)$ such that
\begin{equation}
\max\bigl\{\mathfrak p_\beta(\widetilde{\BBE}),\mathfrak p_\beta(\Bcal_y)\bigr\}\le 
\Bigl(c_4\Delta^{-2}\,\texte^{\beta J c_3 \Delta^3}\Bigr)^{B^2}\frac{B^2\texte^{-c_5\beta J\Delta^2}}{1-\mathfrak p_\beta(\BBE)}
\end{equation}
\end{lemma}

\begin{proof}{Proof (Sketch)}
We will instead estimate the objects $\mathfrak p_\beta(\widetilde{\BBE}\setminus\Bcal_y)$ and $\mathfrak p_\beta(\Bcal_y\setminus\widetilde{\BBE})$; from these the claim will follow by invoking subadditivity $\Acal\mapsto\mathfrak p_\beta(\Acal)$. Consider the partition function $\widetilde Z_{N,\beta}(\widetilde{\BBE}\setminus\Bcal_y)$ and let~$\tilde{\bS}$ be a configuration contributing to the sum. As the second component of all spins are small, the partition function is bounded as in the two-component spin case. The result is
\begin{equation}
\widetilde Z_{N,\beta}(\widetilde{\BBE}\setminus\Bcal_y)
\le \texte^{-\beta J|\T_N|}[2\sqrt{2\pi}]^{|\T_N|}\bigl[B^2\texte^{-\beta J\Delta^2}\bigr]^{(N/B)^2}.
\end{equation}
Combining this with the lower bound on~$Z_{N,\beta}$ from the previous proof now gives
the desired bound for $\mathfrak p_\beta(\widetilde{\BBE}\setminus\Bcal_y)$. 

Concerning the event $\Bcal_y\setminus\widetilde{\BBE}$, as this is disjoint from~$\BBE$, we just note that, in this case, that we can bound
\begin{equation}
\mathcal H_N(\tilde{\bS})\ge\mathcal H_N(\bS)+J\Delta^2\Bigl(\frac NB\Bigr)^2+O(\Delta^3)J|\T_N|
\end{equation}
and then proceed as in the previous lemma. The desired suppression now comes from the second term on the right hand side of the last display.
\end{proof}

We are now ready to prove our main result on quantum systems with large spins:

\begin{proof}{Proof of Theorem~\ref{T:QLRO}}
Key to the formalism of~\cite{Biskup-Chayes-Starr} is the operator $\hat Q_{\Acal}$ associated with the event~$\Acal$ as follows:
\begin{equation}
\hat Q_{\Acal}:=\left(\frac{2\Scal+1}{4\pi}\right)^{\T_N}\int_{\Acal}\textd\Omega\,\,|\Omega\rangle\langle\Omega|.
\end{equation}
Notice that if $\Acal$ is a $B$-block event, then $\hat Q_{\Acal}$ behaves as identity on the part of the Hilbert space outside $\Lambda_B$. The map $\Acal\mapsto\hat Q_{\Acal}$ is countably additive; in particular, if $\Acal_1,\dots,\Acal_n$ form a partition of the probability space, then $\hat Q_{\Acal_1}+\dots+\hat Q_{\Acal_n}=\1$.

Consider now the operators $\hat Q_{\widetilde\Gcal_x}$, $\hat Q_{\widetilde\Gcal_z}$ and $\hat Q_{\Bcal}$. Notice that these operators are invariant under reflection of the box through any mid-plane (because so are the events they arise from). Let $\vartheta_{\bt}$ denote the shift on $\T_N$ by $B\bt$, where $\bt\in\T_{N/B}$. Let $\xi:=\frac c{\Scal}$, where $c$ is as in Lemma~\ref{lemma-q-class}. Since $\beta\le c_1\sqrt{\Scal}$ and $\mathfrak p_\beta(\Bcal)\texte^{\xi+c_2\beta/\sqrt{\Scal}}$ is small --- of course, for~$\beta$ large, thanks to Lemmas~\ref{lemma-5.2} and~\ref{lemma-5.3}, Proposition~\ref{P:smallBcal} and the choices \eqref{E:Delta-B} --- for some absolute constant $c_1,c_2\in(0,\infty)$, Theorem~3.7 of~\cite{Biskup-Chayes-Starr} tells us that for some $\epsilon>0$ small
\begin{equation}
\label{E:B}
\langle\hat Q_{\Bcal}\rangle_{N,\beta}<\epsilon
\end{equation}
and
\begin{equation}
\label{E:Q}
\bigl\langle\hat Q_{\theta_{\bt_1} \widetilde\Gcal_x}(1-\hat Q_{\theta_{\bt_2} \widetilde\Gcal_z})\bigr\rangle_{N,\beta}<\epsilon.
\end{equation}
for any $\bt_1,\bt_2\in\T_{N/B}$.

In order to process these to the bounds in the statement of the theorem, let us note that, if~$\Ocal$ is a bounded operator generated by $\hat S_{\br}^{\alpha}$ with $\alpha=x,y,z$, and $\br\in (\bt B+\Lambda_B)\cup (\bt' B+\Lambda_B)$ and  $[\Ocal]_{\Omega}$ is its upper index for which we set
\begin{equation}
\gamma:=\sup_\Omega\bigl|[\Ocal]_{\Omega}\bigr|
\end{equation}
and
\begin{equation}
\label{E:gamma}
\gamma_\alpha:= \sup\Bigl\{ \bigl|[\Ocal]_{\Omega}\bigr| : \Omega\in \theta_{\bt}\Gcal_{\a} \cup \theta_{\bt'}\Gcal_{\a}  \Bigr\},\qquad\alpha=x,z,
\end{equation}
then \eqref{E:Q} and \eqref{E:gamma} imply
\begin{equation}
\label{E:O<gamma}
\langle \Ocal \rangle_{N,\b}\le 3\epsilon \gamma+\max\{\g_x,\g_z\}.
\end{equation}
Indeed, writing 
\begin{equation}
\label{E:3e+g}
\langle \Ocal \rangle_{N,\b}=\frac1{Z_{N,\b}}\left(\frac{2\Scal+1}{4\pi}\right)^{\abs{\T_N}}\int_{(\mathscr S_2)^{\Lambda}}\textd\Omega\,\,[\Ocal]_\Omega\,\langle\Omega| \texte^{-\b \hat H_N}  |\Omega\rangle,
\end{equation}
we can estimate $[\Ocal]_\Omega$ by~$\gamma_\alpha$ on $\theta_{\bt}\Gcal_{\a} \cup \theta_{\bt'}\Gcal_{\a}$, and by~$\gamma$ on the complement of these. The positivity of $\langle\Omega| \texte^{-\b \hat H_N}  |\Omega\rangle$ permits us to convert the resulting integrals  to expectations of the kind \twoeqref{E:B}{E:Q}.

To see how this applies in a specific situation, consider $\Ocal:= \Scal^{-4}\bigl[\hS_{\br}^{y}\bigr]^2\bigl[\hS_{\br'}^{y}\bigr]^2$ for $\br, \br'$ belonging to distinct translates of $\Lambda_B$. We then get
\begin{equation}
[\Ocal]_\Omega=\bigl[\O_{\br}^{y}\bigr]^2  \bigl[\O_{\br'}^{y}\bigr]^2 + O(\Scal^{-1})   ,
\end{equation}
which is $O(\eta+\Scal^{-1})$ on $\cup_{\a=x,z}(\theta_{\bt}\Gcal_{\a} \cup \theta_{\bt'}\Gcal_{\a}) $ and $O(1)$ otherwise. As a consequence, for some $c<\infty$ and $N$ is sufficiently large,
\begin{equation}
\Scal^{-4}\Bigl|\bigl\langle\bigl[\hS_{\br}^{y}\bigr]^2\bigl[\hS_{\br'}^{y}\bigr]^2\bigr\rangle_{N,\beta}\Bigr|\le c(\eta+\Scal^{-1}+\epsilon).
\end{equation}
The other cases needed to establish \twoeqref{E:QLRO22}{E:QLRO13} are checked analogously.
Once \twoeqref{E:QLRO22}{E:QLRO13} are proved, the bound \eqref{E:QLRO} follows by the symmetries of the model and the fact that $\Scal^{-2} \sum_{\a} \bigl[S_{\br}^{\a}\bigr]^2=1+O(1/\Scal)$.
\end{proof}

\begin{proof}{Proof of Theorem~\ref{T:QLRO2}}
Fix $J>0$ and throughout this proof let $\hHcal_N$ denote the Hamiltonian for $J_1=J_2=J$.
Consider the operator 
\begin{equation}
\hat{\mathcal E}_N^{\alpha}:={\mathcal S}^{-2}\sum_{\begin{array}{c}\scriptstyle\,\langle\br,\br'\rangle_\alpha\\*[-6pt]\scriptstyle\br,\br'\in \T_N\end{array}
} \, S_{\br}^{x}S_{\br'}^{x}, \qquad  \a=x,y,z
\end{equation}
and notice that increasing $J_1$ above the common value $J$ amounts to adding the term 
$(J-J_1)\hat{\mathcal E}_N^{x}$ to $\hHcal_N$. Key to the proof is to show that 
\begin{equation}
f(h):= \lim_{N\to\infty} \frac1{N^2}\log \frac{\Tr(\texte^{-\b\hHcal_N+h \hat{\mathcal E}_N^{x} })}{\Tr ( \texte^{-\b\hHcal_N})}
\end{equation}
satisfies 
\begin{equation}
\label{E:f+}
\frac{\textd}{\textd h^+}f(h)|_{h=0} > 2(1-\delta).
\end{equation}
Indeed, $h\mapsto f(h)$ is convex (by H\"older inequality) and so \eqref{E:f+} implies the same bound for all $h>0$.
This in turn shows that,  for $N$ sufficiently large and $\br$ with both coordinate  even,
\begin{equation}
{\mathcal S}^{-2}\langle \hat{\mathcal E}_{\br}   \rangle_{N,\beta}\ge 2(1-\delta)
\end{equation}
where $\hat{\mathcal E}_{\br}$ is quantum counterpart of the plaquette energy \eqref{E:WJen}. From here the desired claims follow along the same argument as in the classical case.

We thus have to show \eqref{E:f+}. To this end consider the torus events
\begin{equation}
\Acal_\alpha:=\Bigl\{\Omega\colon\sum_{\langle\br,\br'\rangle_\alpha}\Omega_{\br}^\alpha\Omega_{\br'}^\alpha\ge(2-\delta)|\T_N|\Bigr\},\qquad \alpha=x,z.
\end{equation}
A straightforward application of chessboard estimates shows that, for each~$\delta>0$, there is $\beta_1=\beta_1(\delta)$ such that for $\beta\ge\beta_1$ and $\Scal\ge c\beta^2$, we have
\begin{equation}
\label{E:QAbd}
\bigl\langle\hat Q_{\Acal_\alpha}\bigr\rangle_{N,\beta}\ge\frac12-\epsilon_N,\qquad \alpha=x,z,
\end{equation}
where $\epsilon_N\to0$ as $N\to\infty$. Next we note that, by Jensen's inequality
\begin{equation}
\label{E:TR-LB}
\Tr\bigl(\texte^{-\beta\hHcal_N+h\hat{\mathcal E}_N^\alpha}\bigr)
\ge\left(\frac{2\Scal+1}{4\pi}\right)^{\abs{\T_N}}\int_{\Acal_\alpha}\textd\Omega\,\,\texte^{-\beta\langle\hHcal_N\rangle_\Omega+h(2-\delta)|\T_N|},
\end{equation}
where we already applied that $\langle\hat{\mathcal E}_N^\alpha\rangle_\Omega\ge(2-\delta)|\T_N|$ for $\Omega\in\Acal_\alpha$. Theorem~3.1 of~\cite{Biskup-Chayes-Starr} now shows that
\begin{equation}
\texte^{-\beta\langle\hHcal_N\rangle_\Omega}\ge\langle\Omega|\texte^{-\beta\hHcal_N}|\Omega\rangle\,\texte^{-c\beta|\T_N|/\sqrt\Scal}
\end{equation}
for some~$c<\infty$. Substituting this into \eqref{E:TR-LB} yields
\begin{equation}
\frac{\Tr(\texte^{-\beta\hHcal_N+h\hat{\mathcal E}_N^\alpha})}{\Tr(\texte^{-\beta\hHcal_N})}
\ge\texte^{h(2-\delta)|\T_N|-c\beta|\T_N|
/\sqrt\Scal}\bigl\langle\hat Q_{\Acal_\alpha}\bigr\rangle_{N,\beta}.
\end{equation}
Setting $\alpha:=x$ and applying \eqref{E:QAbd}, the bound \eqref{E:f+} follows once $c\beta/\sqrt\Scal<\delta$.
\end{proof}

\section{Concluding remarks}
\label{sec-VI}

\subsection{Ne\'el \emph{vs} orientational order}
As already mentioned, Wenzel and Janke determined in their numerical experiments that the model exhibits a Ne\'el ordering of the plaquette energies; see Figs.~2(a-b) of~\cite{Wenzel:2009p4934}. Explicitly, energy was found low on the $z$-plaquettes and high on the $x$-plaquettes in a sample obtained by multiple updates of a configuration using Glauber dynamics with the Metropolis rule. We wish to point out that this ordering is a direct consequence of the ORLO established rigorously in the present work. Indeed, if the spins are with high probability aligned with $\hate_2$  ---  i.e., pointing north or south  ---  then the $z$-plaquettes will have energy $\mathcal E_{\br}\approx-4$ while the $x$-plaquettes will have energy $\mathcal E_{\br}\approx0$. In addition, the mixed plaquettes will settle at energy $\mathcal E_{\br}\approx-2$. This is consistent with Fig.~2(d) of~\cite{Wenzel:2009p4934}. 

The state with the spins aligned with $\hate_1$ will have the roles of the $x$ and $z$-plaquettes interchanged, giving the energy distribution again a Ne\'el type order resemblance. Notwithstanding, the physical significance of a Ne\'el order is unclear given the period-2 nature of the interaction. And, in fact, matters seem to look quite different when instead of~$\mathcal E_{\br}$ we work with more natural quantity,
\begin{eqnarray}
\nonumber
\widetilde{\mathcal E}_{\br}(\bS):=
\bigl[S_{\br}^{\alpha_1}-S_{\br+\hate_1}^{\alpha_1}\bigr]^2
+\bigl[S_{\br}^{\alpha_2}-S_{\br+\hate_2}^{\alpha_2}\bigr]^2\\
\qquad\qquad\qquad
+\bigl[S_{\br+\hate_1}^{\alpha_3}-S_{\br+\hate_1+\hate_2}^{\alpha_3}\bigr]^2
+\bigl[S_{\br+\hate_2}^{\alpha_4}-S_{\br+\hate_1+\hate_2}^{\alpha_4}\bigr]^2,
\end{eqnarray}
which is the plaquette energy normalized to \emph{vanish} in all ground states. Although we do not see any reason why a \emph{strong} Ne\'el ordering should be exhibited by these plaquette energies, it would be perhaps of some interest to rerun the numerical experiments at higher precision to check this fact numerically.

\subsection{Enhanced sampling}
The samples of actual configurations shown in Fig.2(d) of~\cite{Wenzel:2009p4934} deserve one more comment. A cursory look at the figure reveals some level of orientational order in the $\hate_2$-spin direction  ---  which is consistent with our mathematical results --- but a more careful analysis uncovers an apparent statistical discrepancy. Indeed, most of the spins point down in the figure although the plaquette-flip symmetries $(\bvarphi_{\br})$ of the Hamiltonian, which can be applied independently at all even locations, indicate that about half of the $z$-plaquettes should be pointing up! It is thus somewhat surprising that in Fig.2(d) of~\cite{Wenzel:2009p4934} only 7 such plaquettes out of the total of 25 have upward-pointing spins; the others are clearly pointing down.
As the number of upward-pointing plaquettes is well approximated by a binomial distribution, the probability that this happens is 
\begin{equation}
\biggl(\begin{array}{c}{25}\\7\end{array}\biggr)\frac1{2^{25}}\approx0.014,
\end{equation}
i.e., the configuration in the figure will typically appear only once in about 70 samples! We take this as a possible indication that, despite judicious methods of simulations, the configuration may not have fully equilibrated at the time the snapshot was taken.

Turning this observation into a positive statement, one can try to use the plaquette-spin flips to accelerate the convergence of the computer sample to equilibrium.  Notice that, even at moderate temperatures, any single-spin update rule will have considerable difficulties to overcome the energy barrier associated with changing the orientation of an entire plaquette. The dynamics would naturally mix faster if an occasional flip of an entire plaquette --- by an application of one of the maps $(\bvarphi_{\br})$ --- is incorporated into the stochastic dynamics. This would result in an algorithm reminiscent of the Swendsen-Wang method for sampling configurations in the $q$-state Potts model~\cite{Swendsen-Wang}. It could be expected that this enhancement would result in a substantially better performance of the simulations.
An algorithm of this sort has been recently attempted in the context of the orbital compass model~
\cite{WJL}, although there the cluster flips have to be performed along entire lines of sites which makes them very non-local.

\subsection{Correlation decay}
Our mathematical argument establishes rigorously long-range order in the system. Nonetheless, we do so without giving any bound on the decay of (truncated) correlations. Since our argument is based on contour methods and suppression of long-wavelength part of the spin-wave decomposition, we tend to believe that the correlations generally decay exponentially fast in any translation-invariant, ergodic Gibbs measure for this system. However, we have not been able to find a rigorous argument in the vain. Again, it would be of interest to see if this question could be addressed by numerical methods.

\vskip0.8cm\noindent
\leftline{\bf Acknowledgments}

\bigskip\noindent
The research of M.B. was partially supported by the NSF~grant DMS-0949250. The research of R.K. was partially supported by the grants GA\v CR~201-09-1931 and MSM~0021620845.

\vskip0.8cm\noindent
\leftline{\bf References}

\medskip\noindent

\end{document}